\documentclass[11pt]{amsart}
\usepackage[margin=1in]{geometry}
\usepackage{graphicx}
\usepackage{amsmath}
\usepackage{amsthm}
\usepackage{physics}
\usepackage{amssymb}
\usepackage{xcolor}
\usepackage{hyperref}
\usepackage{mathtools}
\usepackage{comment}
\usepackage[font=footnotesize]{caption}
\usepackage{subcaption}
\usepackage{enumitem}
\usepackage{placeins,bbm}
\usepackage[normalem]{ulem}
\usepackage{booktabs}
\usepackage{mathrsfs}

\hypersetup{
    colorlinks,
    linkcolor={red!50!black},
    citecolor={blue!50!black},
    urlcolor={blue!80!black}
}

\theoremstyle{plain}
\newtheorem{theorem}{Theorem}[section]
\newtheorem{lemma}[theorem]{Lemma}
\newtheorem{proposition}[theorem]{Proposition}
\newtheorem{corollary}[theorem]{Corollary}

\theoremstyle{definition}
\newtheorem{definition}[theorem]{Definition}
\newtheorem{example}[theorem]{Example}

\theoremstyle{remark}

\newtheorem{claim}[theorem]{Claim}

\newtheorem{remark}[theorem]{Remark}
\newtheorem{assumption}{Assumption}

\usepackage{etoolbox}

\makeatletter
\AtEndEnvironment{remark}{\hfill$\triangleleft$}
\makeatother
\makeatletter
\AtEndEnvironment{example}{\hfill$\triangleleft$}
\makeatother

\numberwithin{equation}{section}

\usepackage{dsfont}
\usepackage{mathrsfs}

\newcommand{\indicator}[1]{\mathds{1}_{#1}}

\DeclarePairedDelimiter{\mset}{\langle}{\rangle}
\DeclarePairedDelimiter{\set}{\{}{\}}

\DeclarePairedDelimiter{\brackets}{[}{]}

\renewcommand{\P}{\mathbb{P}}
\newcommand{\R}{\mathbb{R}}

\newcommand{\Z}{\mathbb{Z}}
\newcommand{\bA}{\mathbb{A}}
\newcommand{\bG}{\mathbb{G}}
\newcommand{\tG}{\widetilde{\bG}}
\newcommand{\bH}{\mathbb{H}}
\newcommand{\bK}{\mathbb{K}}
\newcommand{\bL}{\mathbb{L}}
\newcommand{\bS}{\mathbb{S}}
\newcommand{\bZ}{\mathbf{Z}}
\newcommand{\cA}{\mathcal{A}}
\newcommand{\cC}{\mathcal{C}}

\newcommand{\sF}{\mathscr{F}}

\newcommand{\cL}{\mathcal{L}}
\newcommand{\cM}{\mathcal{M}}

\newcommand{\cP}{\mathcal{P}}

\newcommand{\cT}{\mathcal{T}}
\newcommand{\cV}{\mathcal{V}}

\newcommand{\sX}{\mathscr{X}}

\newcommand{\rmd}{\mathrm{d}}
\newcommand{\hc}{\mathrm{hc}}

\newcommand{\ext}[1]{\mathrm{ext}(#1)}

\newcommand{\supp}[1]{\mathrm{supp}(#1)}

\newcommand{\intr}[1]{\mathrm{int}(#1)}
\newcommand{\interior}[2]{\mathrm{int}_{#1}(#2)}

\newcommand{\onion}[2][n]{#2^{(#1)}}
\newcommand{\exb}{\partial^{\mathrm{ex}}}
\newcommand{\remove}[1]{}

\definecolor{newred}{HTML}{ff382e}
\definecolor{newgreen}{HTML}{549641}
\definecolor{newblue}{HTML}{4c4cfc}
\definecolor{neworange}{HTML}{c98702}

\allowdisplaybreaks
\let\emptyset\varnothing

\definecolor{Qcolor}{RGB}{31,120,180}

\title[The maximal hard-core model]{A general framework for crystallization in maximal hard-core models}
\author[A. Barg]{Alexander Barg}\address{(AB) University of Maryland, College Park, MD 20817, USA}\email{abarg@umd.edu} 
\author[Q. He]{Qidong He}\address{(QH) Northeastern University, Boston, MA 02115, USA}\email{q.he@northeastern.edu}
\author[G. Wang]{Geyang Wang}\address{(GW) University of Maryland, College Park, MD 20817, USA}\email{wanggy@umd.edu}

\begin{document}

\begin{abstract}
We study general maximal hard-core models on periodic lattice-type graphs. While for standard models, phase coexistence may arise for high activity, for maximal models this behavior may occur also for sufficiently low activity values. Relying on the concept of volume allocation, we develop a {\em unified set of assumptions} that imply the Peierls condition, a convergent cluster expansion for the
partition function, and hence the conclusions of Pirogov--Sinai theory for this class of models. We further check these assumptions for a number of examples including standard lattice-type periodic graphs, proving crystallization for both high and low activity. We also derive estimates for the values of activity that bound the phase coexistence regions in the phase diagram, rewriting for this purpose the proof of a technical result in the derivation of 
Pirogov--Sinai theory.
\end{abstract}
\maketitle

\section{Introduction}

The hard-core model is a classic object of study in statistical physics and probabilistic combinatorics \cite{gaunt1965hard,dobrushin1968problem,brightwell2002,Jenssen_2024,blanca2019phase,mazel2019high,mazel2025high,he2024high}. The simplest model in this class is the model of independent
sets of vertices in a graph $\bG$, where an independent set $I$ in a finite region $\Lambda\Subset \bG$ is assigned a probability that depends on the number of vertices in the set (`occupied sites') and is proportional to $u^{|I|}$, where $u$ is called the activity parameter. Once this distribution is extended to the entire graph $\bG$, thereby yielding a Gibbs measure, one of the main questions studied in the literature is whether, for a given $u$, this measure is unique or whether multiple such measures exist. 
The value of $u$ at which the answer changes is associated with the phase transition and is called the critical activity. 
The existence of a critical activity is not guaranteed and is not known even in such simple cases as $\bG=\Z^2$, the 
square lattice on the plane. Sometimes it is nevertheless possible to show that for sufficiently large $u$ there are multiple Gibbs measures; for instance for $\Z^2$ this is true starting with $u\ge 5.3506$ \cite{blanca2019phase}.

A variation of the hard-core model on lattices, considered recently in \cite{wang2025maximal,wang2026recoverable}, limits itself to {\em maximal independent sets}, pursuing the same set of questions as in the classical case. Given the significant attention devoted to the latter, it is somewhat surprising that the maximal model has not previously been considered in the mathematical literature (we are aware of only one earlier work, \cite{dallastaStatisticalMechanicsMaximal2009}, which studied this model on random graphs through Monte Carlo simulations). The two cited papers studied phase coexistence for the maximal model for $\bG=\Z^2$ and the planar triangular lattice $\bA_2$. A new feature of
the maximal model compared to the standard case arises because, unlike independent sets, maximal independent sets cannot be too sparse without breaking the maximality constraint, and
therefore it makes sense to address the phase coexistence question both for high- and
low-activity cases. For the low-activity case, the analysis in \cite{wang2025maximal,wang2026recoverable} developed the ideas in \cite{mazel2019high,mazel2025high}, using Pirogov--Sinai theory to draw conclusions
about phase coexistence. 

In this work, we continue the study of the maximal hard-core model on periodic graphs, 
phrasing the problem in more general terms. Namely, we allow several different types of
tiles with distinct chemical potentials. A configuration of the model is formed
once tiles are placed on vertices of $\bG$, as detailed below, under the assumptions 
that tiles placed on distinct vertices do not overlap, forming an independent set (in the corresponding exclusion graph), 
and that no additional tile can be added without creating overlaps, supporting the maximality condition.

Our main goal is to develop a {\em general approach} to phase coexistence and crystallization for
both high- and low-activity regimes relying on Pirogov--Sinai theory. A formal approach
to the application of this theory to the hard-core case was recently expounded in \cite{mazel2024pirogov}, 
where the hard-core constraints were formulated as interaction models with infinite potentials. 
We use this approach in the form presented in \cite{he2026unified}, relying on the 
notion of volume allocations. We further rely on the contour formalism of Zahradn{\'i}k \cite{zahradnik1984alternate} to isolate the assumptions the 
allocation should satisfy to support the Peierls condition on the energy of the contours
that separate distinct ground states in the configuration. Once this condition is satisfied,
the conclusions of Pirogov--Sinai theory apply to the model, implying that stable periodic
ground states give rise to distinct extremal periodic Gibbs measures \cite{zahradnik1984alternate}. We verify these assumptions for a group of 
lattice/periodic graphs frequently studied in the literature, including the hypercubic lattices $\Z^d$ ($d\ge 2$),
the triangular lattice $\bA_2$, and the hexagonal (honeycomb) lattice $\bH_2$ for both low- and high-activity 
regimes; if they are satisfied, crystallization follows immediately from the general results
in this paper.

 Our other goal is to extract conditions on the parameters (the inverse temperature and the Peierls constant) for Pirogov--Sinai theory to apply. For the standard (non-maximal) hard-core
model certain conditions were previously derived in \cite{he2024high}; however, these results are not in a form applicable to the maximal model. To find explicit conditions on the
model parameters, in Section~\ref{sec:PS-theory} we re-derive Zahradn{\'i}k's theorem in the form that implies certain conditions for the parameter regimes that support  the conclusions
of Pirogov--Sinai theory. In examples, however, these conditions yield surprisingly poor estimates of the model parameters, for instance, on the values of high/low activity
starting from which the theory applies. Improving these inequalities is therefore an open problem, which we leave for future research.

\section{The model}
\label{sec:the model}
We begin with a general definition of the hard-core model on lattice (periodic) graphs. 
Let $B$ be a non-degenerate $d\times d$ matrix ($d\ge 2$), let $\bL:=B\cdot\Z^d$ be a lattice, and let $\bG$ be the union of finitely many disjoint translations of $\bL$. 
Let $\rho_1:=\min\set{\abs{u-v}\mid u,v\in\bG,\ u\ne v}$ be the smallest
distance between distinct sites, where $\abs{\cdot}$ denotes the Euclidean norm,
and let $\rho_2>\rho_1$ be the next distinct distance.
We draw an edge between $u,v \in \bG$ if $\abs{u-v}=\rho_1$,
and write $u \sim v$ in that case.
We consider hard-core models defined on graphs $\bG$ of this form.

Examples, normalized to have nearest-neighbor distance one ($\rho_1 = 1$), include:
\begin{itemize}
    \item the hypercubic lattices $\Z^d$ ($d\ge 2$), with $\bG=\bL=\Z^d$;
    \item the triangular lattice $\bA_2$, with
    \begin{equation*}
        \bG=\bL
        =\set{m(1,0)+n(1/2,\sqrt{3}/2)\mid m,n\in\Z};
    \end{equation*}
    \item the honeycomb lattice $\bH_2$, obtained from the scaled triangular lattice $\bL=\sqrt{3}\bA_2$ as the disjoint union $\bG=\bL\sqcup\left((0,1)+\bL\right)$.
\end{itemize}
See~\cite{grunbaum1987tilings} for these and other periodic planar lattices.

\begin{figure}[t]
    \centering
    \includegraphics[width=0.94\textwidth]{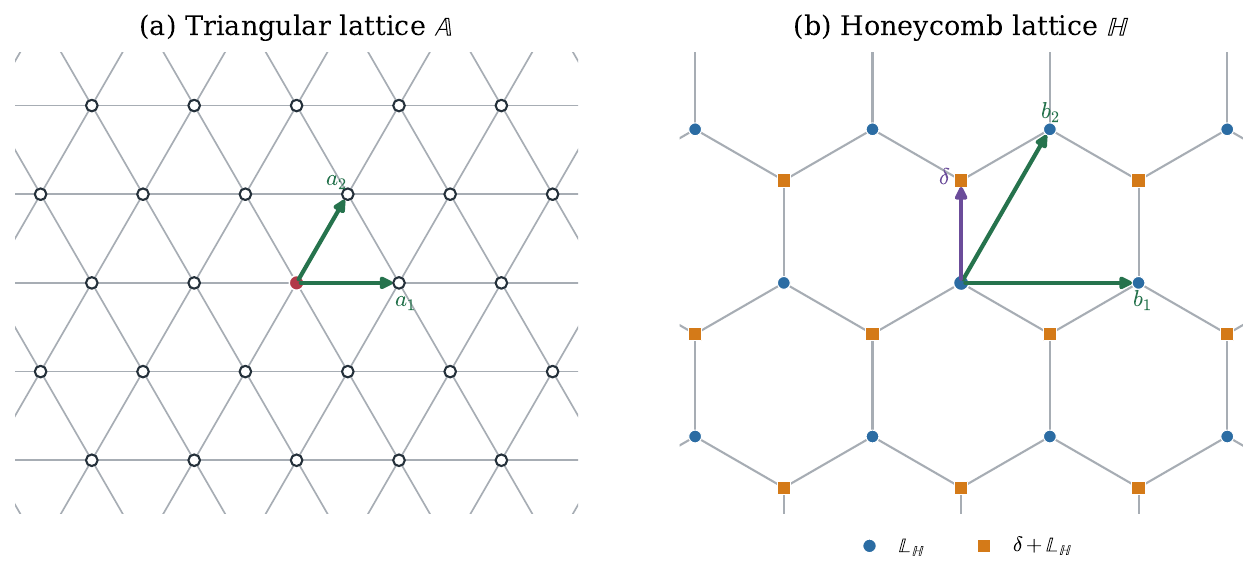}
    \caption{The triangular and honeycomb lattices in the coordinate
    realizations above.  Left: $\bA_2$ is generated by
    $a_1=(1,0)$ and $a_2=(1/2,\sqrt3/2)$.  Right:
    $\bH_2=\bL\sqcup(\delta+\bL)$, where
    $\bL$ is generated by $b_1=(\sqrt3,0)$ and
    $b_2=(\sqrt3/2,3/2)$, and $\delta=(0,1)$.  Circles and squares
    distinguish the two translates.  Edges are drawn between nearest neighbors,
    all at distance one.}
    \label{fig:triangular and honeycomb lattices}
\end{figure}

Throughout the paper,
$B_r(s):=\set{s'\in\R^d\mid\abs{s'-s}<r}$ is the open ball in $\R^d$ of radius $r$.
For the nearest-neighbor graph defined above, we write $\rmd_{\mathrm{gr}}(u,v)$
for the \textbf{graph distance} between $u,v\in\bG$, that is, the number of
edges of a shortest path joining them, with $\rmd_{\mathrm{gr}}(u,v)=\infty$
when no such path exists, and
\begin{equation*}
    \label{eqn:graph ball}
    B^{\mathrm{gr}}_r(v):=\set{u\in\bG\mid \rmd_{\mathrm{gr}}(u,v)\le r},
    \qquad r\in\Z_{\ge 0},
\end{equation*}
for the corresponding \textbf{graph ball} of radius $r$ around $v$. 

Let $T_1,\dots,T_k$ be non-empty, bounded subsets of $\R^d$, which we refer to as \textbf{tiles}. 
A \textbf{hard-core configuration} is a function $\omega:\bG\to\set{0,1,\dots,k}$ such that 
\begin{equation}\label{eq:hard-core}
    \intr{u+T_{\omega(u)}}\cap\intr{v+T_{\omega(v)}}=\emptyset\quad
\end{equation}
for every distinct $u,v\in\bG$ such that $\omega(u),\omega(v)\ne 0,$ where $\intr{\cdot}$ denotes the interior of a set. 
We identify a hard-core configuration $\omega$ with the set 
\begin{equation}
    \label{eqn: set interpretation of configuration}
    X=X(\omega):=\set{(v,\omega(v))\mid v\in\bG,\omega(v)\ne 0},
\end{equation}
and write $X_\Lambda:=X\cap(\Lambda\times\set{1,\dots,k})$ for $\Lambda\subset\R^d$, and $T_x:=v+T_{\omega(v)}$. Fix a \emph{unit} vector $\vec{\mu}=(\mu_1,\dots,\mu_k)\in\bS^{k-1}\subset\R^k$, whose components are interpreted as the chemical potentials of the tiles and let  $\mu_x:=\mu_{\omega(v)}$ for $x=(v,\omega(v))\in X(\omega)$. We also call $x=(v,\omega(v))\in X(\omega)$ a \textbf{particle} of $X$.
A hard-core configuration $X$ is \textbf{maximal} if for all $(v,i)\in\bG\times\set{1,\dots,k}$, either $(v,i)\in X$ or $X\cup\set{(v,i)}$ is not a hard-core configuration in the sense that adding $(v,i)$ to $X$ would violate \eqref{eq:hard-core} or $\omega(v)\ne 0$. In other words, $X$ is maximal if it is a maximal element in the set of hard-core configurations with respect to set containment.
Let $\Omega$ be the set of maximal configurations and equip it with the (metrizable) product topology~\cite[Section~4.1]{georgii2011gibbs}: Under this topology, two maximal configurations are \emph{close} if they agree on a large ball centered at the origin. 

\begin{remark} To connect this definition with the hard-core model of independent sets ubiquitous in combinatorics, recall the definition of the hard-core Gibbs measure (with or without the maximality constraint, which is inconsequential in this discussion). Let $\sX(\bG)$ be the collection of independent sets in $\bG$. For $X\in\sX(\bG)$ let $X_\Lambda$ be its restriction to a finite region $\Lambda\Subset \bG$. A probability measure $\P$ on $\sX(\bG)$ is called a 
{\em hard-core Gibbs measure} with {\em activity} $\sf u$  if for every $\Lambda\Subset \bG$, every $X'\in \sX(\bG)$ and $\P$-almost every $X\in \sX(\bG)$, the conditional probability on $\Lambda$ satisfies
\begin{equation} \label{eq: hardcore measure}
  \P_\Lambda(X'\mid X)\propto 
  {\sf u}^{\abs{X'_{\Lambda}}}\indicator{\{X'_{\Lambda^c}=X_{\Lambda^c}\}},
\end{equation}
where the indicator isolates the event that $X'$ and $X$ coincide on the {\em complement} of $\Lambda$ in $\bG$ and where the right-hand side should be appropriately normalized. This measure can be extended to the entire lattice $\bG$ via the usual Dobrushin--Lanford--Ruelle (DLR) prescription~\cite{biskup2009reflection,georgii2011gibbs}.

To link our definition with this model, we limit ourselves to particles of a single type, taking $k=1$. Then $\mu_1\in\set{\pm1}$, and the activity is controlled by the additional parameter, $\beta$, which plays the role of inverse temperature.  Taking ${\sf u}=e^{\beta\mu_1}$, we recover the model of independent sets with Hamiltonian written as in~\eqref{eqn:formal Hamiltonian} below. 
\end{remark}

\begin{remark}
    \label{rem:maximality is local}
    We note that maximality can be verified by checking a set of local conditions, since a hard-core configuration $X$ fails to be maximal if and only if there exist $v\in\bG$ and $1\le i\le k$ such that $(v,i)\not\in X$ and $X\cup\set{(v,i)}$ is also a hard-core configuration. 
    \end{remark}

The formal Hamiltonian of a maximal configuration $X\in\Omega$ is defined by
\begin{equation}
    \label{eqn:formal Hamiltonian}
    H(X):=-\sum_{x\in X}\mu_x
\end{equation}
and adapted to finite volumes in the usual way~\cite[Definition (6.14)]{friedli2017statistical}.
The Gibbs measure(s) of the model~\eqref{eqn:formal Hamiltonian} at inverse temperature $\beta$ is/are defined via the DLR prescription.

In this paper, we assume that $\mu_1,\dots,\mu_k$ are either all in $\R_{>0}$ or all in $\R_{<0}$ and focus on the regime $\beta\gg 1$.

The following class of maximal hard-core models is of particular interest.
\begin{definition}[Nearest-neighbor maximal hard-core model]
    \label{def:nearest-neighbor model}
    A maximal hard-core model on $\bG$ is called \textbf{nearest-neighbor} if $k=1$ and the single tile $T$ satisfies
    \begin{equation*}
        \label{eqn:nearest-neighbor tile}
        \intr{u+T}\cap\intr{v+T}\ne\emptyset
        \iff u\sim v
        \qquad\text{for all distinct }u,v\in\bG.
    \end{equation*}
Therefore, $\Omega$ is the collection of maximal independent sets of $\bG$.
Since $k=1$, we identify a particle $x=(v,1)\in X$ with its site
$v\in\bG$, and $\mu_x=\mu_1\in\set{\pm1}$.
One can take $T$ to be a Euclidean ball with radius $r$ satisfying $\rho_1 / 2 < r < \rho_2 /2$.
\end{definition}

\subsection{The assumptions}
\label{sec:the assumptions}

In this subsection, we detail the geometric assumptions that will enable us to control the applicability of Pirogov--Sinai theory for maximal hard-core models. 
In particular, the framework of volume allocations is imported from~\cite{he2026unified}, with suitable modifications (see Items~\ref{itm:continuity} and~\ref{itm:bounded support} of Definition~\ref{def:volume allocation} and Assumption~\ref{asm:optimization}) to accommodate the maximality constraint and chemical potentials of either sign, the latter needed to cover both high- and low-activity regimes.

We begin with the definition of a volume allocation, which
may look somewhat abstract at first glance; the reader is invited to consult Definition~\ref{def: discrete allocation} and Example~\ref{ex:Z2 counting allocation} for a visual perspective.

\begin{definition}
    \label{def:volume allocation}
    Let $S\subset\R^d$ be an $\bL$-invariant space with an $\bL$-invariant, locally finite measure $\lambda$.
    A \textbf{volume allocation} is a family of measurable functions $\phi_x(\cdot\mid X):S\to[0,1]$, defined for every maximal configuration $X\in\Omega$ and $x\in X$, satisfying the following properties:
    \begin{enumerate}
        \item \textit{Translation invariance}: For every $t\in\bL$, 
        \begin{equation}
            \label{eqn:translation invariance}
            \phi_{t+x}(t+\cdot\mid t+X)=\phi_x(\cdot\mid X)\quad\lambda\text{-a.e.},
        \end{equation}
        where $t+X:=\set{(t+v,\omega(v))\mid v\in\bG,\,\omega(v)\ne 0}$ (see~\eqref{eqn: set interpretation of configuration}).
        \item \label{itm:continuity} \textit{Continuity}: For all sequences of maximal configurations $(X_n)_{n\ge 1}$ converging to $X$ in the product topology,
        \begin{equation}
            \label{eqn:continuity}
            \lim_{n\to\infty}\phi_x(\cdot\mid X_n)=\phi_x(\cdot\mid X)\quad\lambda\text{-a.e.}
        \end{equation}
        \item \textit{Partition of unity}: 
        \begin{equation}
            \label{eqn:partition of unity}
            \sum_{x\in X}\phi_x(\cdot\mid X)=1\quad\lambda\text{-a.e.}
        \end{equation}
        \item \label{itm:bounded support} \textit{Bounded support}: There exists a constant $R>0$ such that for every maximal configuration $X \in \Omega$ and $x\in X$, $\set{s\in S\mid\phi_x(s\mid X)\ne 0}\subset B_R(x) \cap S$.
    \end{enumerate}
\end{definition}

Given a volume allocation $\set{\phi_x(\cdot\mid X)}$, the \textbf{allocated volume} of $x\in X$ in a maximal configuration $X\in\Omega$ is defined as
\begin{equation*}
    v(x\mid X):=\int_S\lambda(\dd{s})\phi_x(s\mid X).
\end{equation*}

We introduce the following shorthand notation: For every maximal configuration $X\in\Omega$ and every subset $X'\subset X$ of its particles\footnote{$X'$ is not necessarily maximal.}, write
\begin{equation*}
    v(X'\mid X):=\sum_{x\in X'}v(x\mid X).
\end{equation*}

\begin{definition}
    \label{def: discrete allocation}
    Let $S = \bG$ and $\lambda$ be the counting measure on $S$.\footnote{For the counting measure, $\lambda$-a.e. means for every vertex
    $y\in\bG$.}
    For the nearest-neighbor maximal hard-core model, define the \textbf{discrete volume allocation} by
     \begin{equation}
        \label{eqn:Z2 counting allocation}
        \phi_x(y\mid X):=
        \begin{cases}
            1, & y=x,\\
            1/k_y(X), & y\notin X\text{ and }y\sim x,\\
            0, & \text{otherwise},
        \end{cases}
    \end{equation}
    where, for $y\in \bG\setminus X$,
    \begin{equation*}
        k_y(X):=\abs{\set{x'\in X\mid x'\sim y}}.
    \end{equation*}
    In words, every vertex carries one unit of volume: an occupied vertex keeps
    its own unit, while a vacant vertex divides its unit equally among its
    occupied neighbors.
\end{definition}

    The allocation defined by~\eqref{eqn:Z2 counting allocation} is translation invariant and continuous in the product topology, and satisfies $ \sum_{x\in X}\phi_x(y\mid X)=1$ for every $y \in \bG$.
    Finally, for every $x \in X$, $\{s \in \bG \mid \phi_x(s \mid X) \ne 0\} \subset B_{2\rho_1}(x) \cap \bG$.
    Figure~\ref{fig:Z2 checkerboard allocation} shows an example of~\eqref{eqn:Z2 counting allocation} for the nearest-neighbor maximal hard-core model on $\Z^2$.

    \begin{figure}[tbp]
    \centering
    \includegraphics[width=0.35\textwidth]{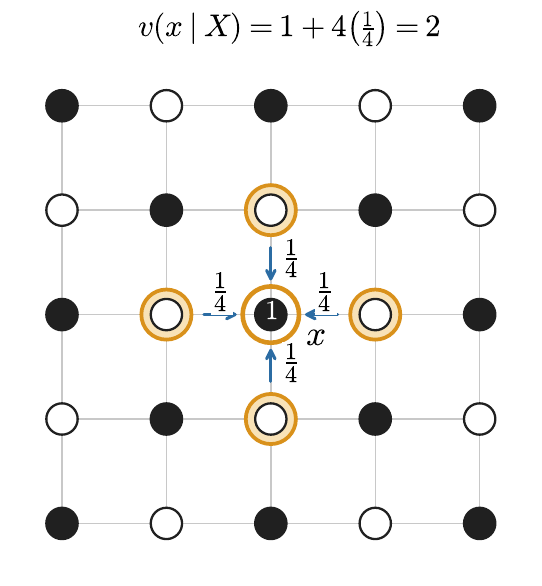}
    \caption{
    Occupied vertices are black and vacant vertices are
    white.  
    The particle $x$ keeps the unit at its own site
    and receives a share $1/4$
    from each highlighted neighbor, because each such neighbor has four occupied
    neighbors. }
    \label{fig:Z2 checkerboard allocation}
    \end{figure}

\begin{assumption}
    \label{asm:optimization}
    Let $\sigma\in\set{\pm1}$ be the (common) sign of $\mu_1,\dots,\mu_k$.
    Assume that there exist a volume allocation $\set{\phi_x(\cdot\mid X)}$ and a constant $p_\ast>0$ such that:
    \begin{enumerate}[label=(\alph*),ref=\theassumption(\alph*)]
        \item \label{itm:minorization} \textit{Local optimization}: For every maximal configuration $X\in\Omega$ and $x\in X$, 
        \begin{equation}
            \label{eqn:local optimization}
            \sigma p_\ast\, v(x\mid X)\ge \mu_x.
        \end{equation}
        \item \label{itm:global optimization} \textit{Global optimization}: There exists a finite number (one or more) of maximal configurations $\overline{X}\in\Omega$, called \textbf{perfect configurations}, such that
        \begin{equation}
            \label{eqn:global optimization}
            \sigma p_\ast\, v(x\mid\overline{X})=\mu_x\quad\text{for every }x\in \overline{X}.
        \end{equation}
    \end{enumerate}
\end{assumption}

\begin{remark}
    The volume allocation bears some similarity to $m$-potentials~\cite{holsztynski1978peierls} that can be used to identify the periodic ground states of the model. We point out that this analogy does not imply a full equivalence, as demonstrated, for instance, by the inequality in~\eqref{eqn:local optimization}. 
    
\end{remark}

\begin{example}
    \label{ex:Z2 counting allocation}
    Consider the nearest-neighbor maximal hard-core model on $\Z^2$ at high activity, for which we take $\sigma=1$ and $\mu_x = 1$. 
    Every occupied vertex $x$ has four vacant neighbors, and $k_y(X)\le 4$ for each vacant neighbor $y$ of $x$.  
    Hence,
    \begin{equation*}
        v(x\mid X)
        =1+\sum_{y\sim x}\frac{1}{k_y(X)}
        \ge 1+4\left(\frac14\right)=2.
    \end{equation*}
    Thus~\eqref{eqn:local optimization} reads
    $\frac12v(x\mid X)\ge1$ with $p_\ast=1/2$.  
    Equality holds for every particle when every vacant neighbor has all four of its neighbors occupied.  This condition yields the two checkerboard configurations, which are the perfect configurations in~\eqref{eqn:global optimization}.

    In comparison, in the low-activity case, we take $\sigma=-1$ and $\mu_x=-1$.
    We have that $v(x\mid X)\le5$ since $k_y(X)\ge1$. 
    There are ten perfect configurations~\cite{golomb1970tiling,jauslin2018high,wang2025maximal}, and $p_\ast = 1/5$.
\end{example}

To state the next assumption, we first introduce some terminology.
A maximal configuration is \textbf{periodic} if it is invariant under translations by a full-rank sublattice of $\bL$.
Let $\bK\subset\bL$ be a full-rank sublattice.
A \textbf{supercell} of $\bK$ is the region $A\cdot[0,1)^d\subset\R^d$, where $A\cdot\Z^d\subset\bK$ is a full-rank sublattice of $\bK$.
Moving in the reverse direction, given a non-degenerate parallelepiped $G=A\cdot[0,1)^d$, let  $\cL^{G}:=A\cdot\Z^d$ be the \textbf{grid} of $G$ (the lattice formed by translations of $G$ by integer multiples of its sides), and denote by $\cT^G$ the set of translations of $G$ under $\cL^G$ (the set of all ``faces'' of $\cL^G$). Let $\sF^G$ be the $\sigma$-algebra generated by $\cT^G$. 
A supercell of a periodic configuration $X$ is a supercell of the largest sublattice of $\bL$ under which $X$ is invariant.
A \textbf{fundamental parallelepiped} of $\bK$ is a supercell of $\bK$ of minimum Euclidean volume.
Figure~\ref{fig:honeycomb supercell} illustrates these notions for a periodic configuration of the honeycomb lattice.

\begin{figure}[htbp]
    \centering
    \includegraphics[width=0.7\textwidth]{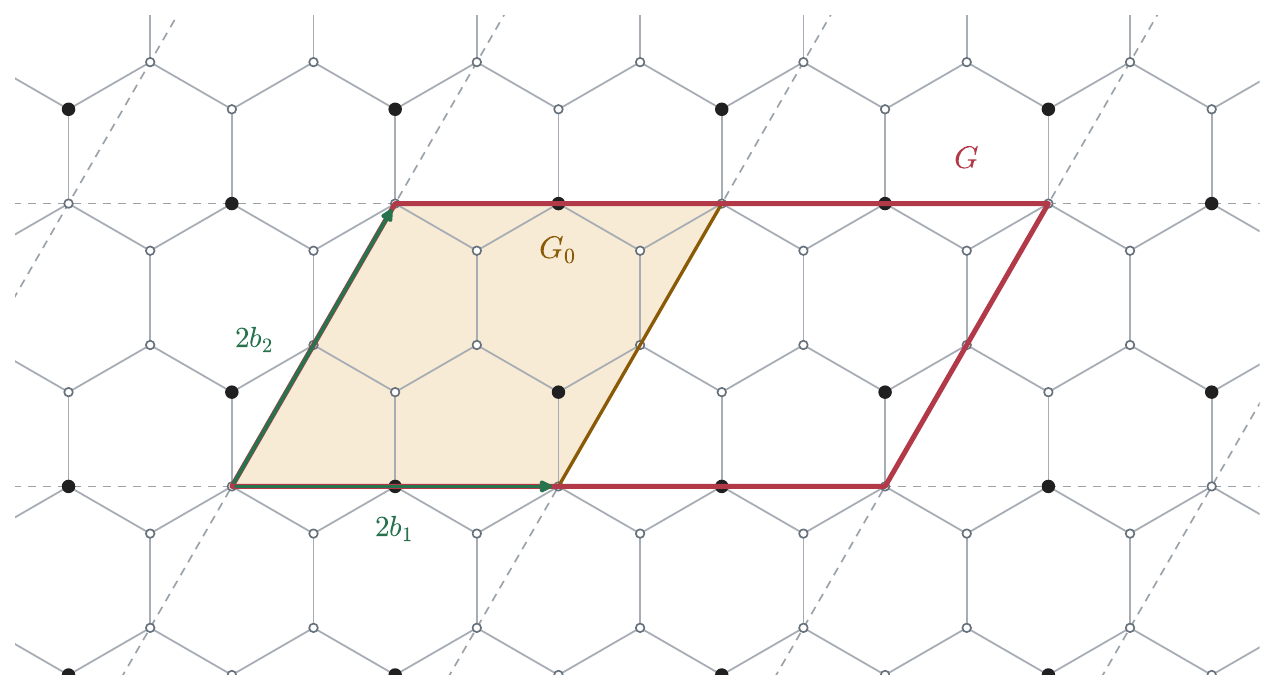}
    \caption{Supercells and fundamental parallelepipeds on the honeycomb
    lattice $\bH_2$, in the coordinates of
    Figure~\ref{fig:triangular and honeycomb lattices}. 
    The configuration is invariant
    under the sublattice $\bK$ generated by $2b_1$ and $2b_2$, which has index $4$
    in $\bL$.  The shaded parallelogram $G_0$ spanned by $2b_1$
    and $2b_2$ is a fundamental parallelepiped of $\bK$, and the
    dashed lines show its translates $\cT^{G_0}$ under the grid
    $\cL^{G_0}=\bK$.  The parallelogram $G$ spanned by $4b_1$ and $2b_2$ is
    another supercell of $\bK$: its grid $\cL^{G}$ has index $2$ in $\bK$, so
    $G$ has twice the minimum volume.}
    \label{fig:honeycomb supercell}
\end{figure}

By Proposition~\ref{prop:perfect iff periodic ground state} below, a perfect configuration is necessarily periodic. Below we will assume (Assumption~\ref{asm:supercell}) that there is
a supercell $F\subset\R^d$ common to all perfect configurations and use the
shorthand $\sF:=\sF^F$. If there is no ambiguity, we refer to the elements of $\cT^F$ as \textbf{cells}. Given $\Lambda\in\sF$, let $\partial^{\mathrm{in}}\Lambda$ denote the union of those cells in $\Lambda$ whose closure intersects the closure of $\Lambda^c$.
Given $G=A\cdot[0,1)^d$, a cell $K=t+G\in\cT^G$ where $t\in\cL^G$,  and an integer $n\ge0$,
define
\begin{equation*}
    \onion[n]{K}
    :=\bigcup_{m\in\set{-n,\ldots,n}^d}(K+A\cdot m).
\end{equation*}
Thus, $\onion[n]{K}$ is the block of $(2n+1)^d$ cells centered at $K$;
in particular, $\onion[0]{K}=K$.

We make the following simple but important observation: 
The maximal hard-core model can be viewed as a finite-range model on $\cL^F\cong\Z^d$, and each maximal configuration can therefore be viewed as an (admissible) spin configuration on $\Z^d$ in the sense of~\cite{mazel2024pirogov}. To justify this, let $F= A \cdot [0,1)^d$ be large enough that a particle (tile) centered in $F$ can interact only with particles centered in $\onion[1]{F}$.
Namely, if $\Lambda,\Lambda'\in\sF$ have disjoint closures and $X,X'$ are hard-core configurations, then $X_\Lambda\cup X'_{\Lambda'}$ is again a hard-core configuration. 
The translates of $F$ partition $\R^d$ and can be indexed by $n \in \Z^d$, and below we use the notation $K_n = A \cdot n + F$. 
Regard the local particle pattern $X_{K_n}$ as a single spin at $n\in\Z^d$.  
The resulting spin space is finite, and we assign energy from~\eqref{eqn:formal Hamiltonian} to each supercell. 
In addition, if either the hard-core or maximality condition is violated between adjacent supercells, we assign $+\infty$ to these cells. 

The notion of \textbf{$\overline{X}$-correctness} is defined in the usual way for every periodic ground state  $\overline{X}\in\Omega$: A cell $K\in\cT^F$ is $\overline{X}$-correct with respect to a maximal configuration $X$ if $X_{\onion[1]{K}}=\overline{X}_{\onion[1]{K}}$.

\begin{assumption}[{Allocation screening by a perfect configuration}]
    \label{asm:supercell}
    There exists a choice of a supercell $F$ common to all \textit{perfect} configurations (see Proposition~\ref{prop:perfect iff periodic ground state}) and larger than the interaction radius with the following property:
If $\overline{X}$ is a perfect configuration and $F$ is $\overline{X}$-correct with respect to a maximal configuration $X$, then for every $x\in X_F$,
    \begin{equation}
        \label{eqn:allocation screening}
        \phi_x(\cdot\mid X)=\phi_x(\cdot\mid\overline{X})\quad\lambda\text{-a.e.}
    \end{equation}
\end{assumption}

In words,~\eqref{eqn:allocation screening} means that if the
central cell $F$ and its surrounding layer agree with a perfect
configuration $\overline{X}$, then changes farther away cannot affect the
allocation of any particle (tile) centered in $F$.  The correct surrounding
layer therefore acts as a buffer between the allocation inside $F$ and
defects outside $\onion[1]{F}$.

\subsection{Periodic ground states}

A maximal configuration ${X}\in \Omega$ is called a \textbf{ground state} of the system if for every bounded region $\Lambda\subset\bG$ and every configuration $X'\in\Omega$ that agrees with $X$ on $\Lambda^c$,
    \begin{equation}\label{eq: periodic ground state}
        H_{\Lambda}(X'_\Lambda  X_{\Lambda^c})-H_{\Lambda}(X_\Lambda X_{\Lambda^c})\ge 0,
    \end{equation}
where $H_{\Lambda}( X'_\Lambda  X_{\Lambda^c}) := -\sum_{x \in X'_{\Lambda}} \mu_x$ is the finite-volume Hamiltonian. Note that even though $X_{\Lambda^c}$ does not explicitly enter the expression for $H_{\Lambda}$, $X_\Lambda'\cup X_\Lambda^c$ is a maximal configuration, which requires accounting for the boundary condition.
In this section, we prove that the periodic ground states are exactly the perfect configurations under Assumptions~\ref{asm:optimization} and~\ref{asm:supercell}.

\begin{proposition}
    \label{prop:perfect iff periodic ground state}
    Under Assumptions~\ref{asm:optimization} and~\ref{asm:supercell}, a maximal configuration $\overline{X}$ is perfect if and only if it is a periodic ground state.
\end{proposition}

First, we note that Properties~\eqref{eqn:translation invariance} and~\eqref{eqn:continuity} of the volume allocation extend to the allocated volume.

\begin{lemma}
    \label{lem:translation invariance and lower semicontinuity}
    The allocated volume is translation invariant and continuous, i.e., for all maximal configurations $X\in\Omega$, $x\in X$, and $t\in\bL$,
    \begin{equation*}
        v(t+x\mid t+X)=v(x\mid X),
    \end{equation*}
    and for all sequences of maximal configurations $(X_n)_{n\ge 1}$ converging to $X$ (in the product topology),
    \begin{equation*}
        \lim_{n\to\infty}v(x\mid X_n)=v(x\mid X).
    \end{equation*}
\end{lemma}

\begin{proof}
    For translation invariance, use the $\bL$-invariance of $S$ and
    $\lambda$, followed by~\eqref{eqn:translation invariance}, to obtain
    \begin{equation*}
        \begin{aligned}
            v(t+x\mid t+X)
            &=\int_S\lambda(\dd{s})\,\phi_{t+x}(s\mid t+X)\\
            &=\int_S\lambda(\dd{u})\,\phi_{t+x}(t+u\mid t+X)\\
            &=\int_S\lambda(\dd{u})\,\phi_x(u\mid X)
            =v(x\mid X),
        \end{aligned}
    \end{equation*}
where $s = t+u$.
    To prove continuity, fix $x\in X$ and let $X_n\to X$ in the
product topology. Then $x\in X_n$ for all sufficiently large $n$;
we restrict attention to these indices. By~\eqref{eqn:continuity},
\[
    \phi_x(s\mid X_n)\longrightarrow\phi_x(s\mid X)
    \qquad\text{for $\lambda$-almost every }s\in S.
\]
Moreover, since the allocation takes values in $[0,1]$,
Item~\ref{itm:bounded support} of
Definition~\ref{def:volume allocation} gives
\[
    0\le\phi_x(s\mid X_n)
    \le\indicator{S\cap B_R(x)}(s).
\]
Since $\lambda$ is locally finite, we have
\[
    \int_S \indicator{S\cap B_R(x)}(s)\,\lambda(\dd{s})
    =\lambda(S\cap B_R(x))<\infty.
\]
Therefore, by the dominated convergence theorem,
\begin{equation*}
    \begin{aligned}[b]
        \lim_{n\to\infty}v(x\mid X_n)
        &=\lim_{n\to\infty}
          \int_S\phi_x(s\mid X_n)\,\lambda(\dd{s})\\
        &=\int_S\phi_x(s\mid X)\,\lambda(\dd{s})\\
        &=v(x\mid X).
    \end{aligned}
    \qedhere
\end{equation*}
\end{proof}

Next, we show that two periodic maximal configurations have the same allocated volume within a common supercell.

\begin{lemma}
    \label{lem:common supercell equal volume}
    If $X,Y$ are periodic maximal configurations and $G$ is a supercell common to both $X$ and $Y$, then
    \begin{equation*}
        v(X_G\mid X)=v(Y_G\mid Y)=\lambda(G\cap S).
    \end{equation*}
\end{lemma}

\begin{proof}
    We first prove the identity for $X$.  By definition,
    \begin{equation*}
        v(X_G\mid X)
        =\sum_{x\in X_G}\int_S\lambda(\dd{s})\,\phi_x(s\mid X).
    \end{equation*}
    Since the collection of sets $\set{(t+G)\cap S\mid t\in\cL^G}$ partitions $S$, we can further write
    \begin{equation*}
        v(X_G\mid X)
        =\sum_{x\in X_G}\sum_{t\in\cL^G}
        \int_{(t+G)\cap S}\lambda(\dd{s})\,\phi_x(s\mid X).
    \end{equation*}
    By the $\bL$-invariaznce of $S$ and $\lambda$ and the translation invariance of the allocation~\eqref{eqn:translation invariance}, using the change of variables $s=t+u$, we obtain
    \begin{equation*}
        \int_{(t+G)\cap S}\lambda(\dd{s})\,\phi_x(s\mid X)
        =
        \int_{G\cap S}\lambda(\dd{u})\,
        \phi_{x-t}(u\mid -t+X) = 
        \int_{G \cap S} \lambda(\dd{u}) \phi_{x} (u \mid X),
    \end{equation*}
    where the last equality follows since $t \in \cL^G$ is a period of $X$.
    Now by the definition of allocated volume and~\eqref{eqn:partition of unity},
    \begin{equation*}
        v(X_G\mid X)
        =\int_{G\cap S}\lambda(\dd{u})
        \sum_{z\in X}\phi_z(u\mid X) 
        =\int_{G\cap S}\lambda(\dd{u})
        =\lambda(G\cap S).
    \end{equation*}
    The same argument gives $v(Y_G\mid Y)=\lambda(G\cap S)$, which proves the claim.
\end{proof}

\begin{remark}
To see a concrete example of Lemma~\ref{lem:common supercell equal volume}, consider the $\Z^2$ checkerboard configuration (see Figure~\ref{fig:Z2 checkerboard allocation}) and the discrete volume allocation, where we take $G=[0,2)\times[0,2)$.
It contains four lattice vertices, so the counting measure gives $\lambda(G\cap\mathbb Z^2)=4$.
There are two tiles in $G$, and each has allocated volume $2$. Hence, $v(X_G\mid X)=2+2=4$.
\end{remark}

We now establish an important screening property of the allocated volume, 
which says that local rearrangements redistribute volume among particles inside $\Lambda$, but the total volume inside $\Lambda$ is unchanged.

\begin{lemma}
    \label{lem:allocated volume screening}
    For every bounded, $\sF$-measurable region $\Lambda\subset\R^d$, perfect configuration $\overline{X}$ and maximal configuration $X$ coinciding with $\overline{X}$ on $\partial^{\mathrm{in}}\Lambda\cup\Lambda^c$,
    \begin{equation*}
        v(X_\Lambda\mid X)=v(\overline{X}_\Lambda\mid\overline{X}).
    \end{equation*}
\end{lemma}

\begin{proof}
    By~\eqref{eqn:partition of unity},
    \begin{equation}
        \label{eqn:lem:allocated volume screening - splitting}
        \sum_{x\in X_\Lambda}\phi_x(\cdot\mid X)
        +\sum_{x\in X_{\Lambda^c}}\phi_x(\cdot\mid X)
        =1
        =\sum_{x\in\overline{X}_\Lambda}\phi_x(\cdot\mid\overline{X})
        +\sum_{x\in\overline{X}_{\Lambda^c}}\phi_x(\cdot\mid\overline{X})
        \quad\lambda\text{-a.e.}
    \end{equation}
    By Assumption~\ref{asm:supercell}, the second sums on both sides of~\eqref{eqn:lem:allocated volume screening - splitting} are $\lambda$-a.e. equal.
    After cancellation, integrating both sides over $S$ yields the lemma.
\end{proof}

In the next lemma, $\rmd_\infty$ is the distance on $\cL^G$ measured in the units of the basis vectors, where the basis is obtained from the standard basis of $\Z^d$ upon the transformation $\cL^G=A\cdot \Z^d$.

\begin{lemma}
    \label{lem:maximal completion property}
For every $X,Y\in \Omega$ with a common supercell $G$ large enough so that any particle centered in $G$ can interact only with particles centered in $\onion[1]{G}$, and for every $\Lambda\subset \cL^G$, there is a $Z\in\Omega$
    such that 
    \begin{equation}\label{eq:richness}
            Z_\Lambda=X_\Lambda,
            \qquad
            Z_{\Lambda^c\setminus\partial^{(3)}\Lambda^c}=Y_{\Lambda^c\setminus\partial^{(3)}\Lambda^c},
        \end{equation}
        where $\partial^{(n)}\Lambda^c
    :=\set{K\in \Lambda^c \mid \rmd_\infty(K,\Lambda)\leq n}$.
\end{lemma}

\begin{proof}
    We first observe that every hard-core configuration $V$ on a region
    $\cL^G\cong\Z^d$ can be extended to one that is maximal on $\cL^G$.  Since
    $\cL^G$ is countable and there are finitely many tile types, the possible
    particles $p=(v,i)$ centered in the cells indexed by $\cL^G$ form a countable set.  
    Enumerate those not in $V$ as $p_1,p_2,\ldots$.  
    Starting from $V_0=V$, define successively
    \begin{equation*}
        V_j:=
        \begin{cases}
            V_{j-1}\cup\set{p_j},&\text{if this union is a hard-core
            configuration},\\
            V_{j-1},&\text{otherwise},
        \end{cases}
        \qquad \text{ and }
        \widehat V:=\bigcup_jV_j.
    \end{equation*}
    The increasing union $\widehat V$ is a hard-core configuration.  
    Every particle $p_k$ missing from $\widehat V$ was rejected at a finite stage $k$ because
    adding it would violate the hard-core condition with a particle already
    chosen, which remains in $\widehat V$.  Hence $\widehat V$ is maximal on
    $\cL^G$.  

    Now we take $V$ to be the union $X_{\Lambda \cup \partial^{(1)} \Lambda^c} \cup Y_{\Lambda^c\setminus\partial^{(2)}\Lambda^c}$, which is a hard-core configuration that can be extended to a maximal configuration $\widehat{V}$ as described above.
    Note that $\widehat{V}_{\Lambda} = X_\Lambda$ and $\widehat{V}_{\Lambda^c\setminus\partial^{(3)}\Lambda^c} = Y_{\Lambda^c\setminus\partial^{(3)}\Lambda^c}$ since all newly added particles are contained in $\partial^{(3)} \Lambda^c$. 
    Therefore, taking $Z = \widehat{V}$ completes the proof.
\end{proof}

\begin{proof}[Proof of Proposition~\ref{prop:perfect iff periodic ground state}]
    By Lemma~\ref{lem:translation invariance and lower semicontinuity}, every translation of a perfect configuration is again a perfect configuration.
    Since there are only finitely many perfect configurations, each perfect configuration must be periodic.
    To see that a perfect configuration $\overline{X}$ is a ground state, let $\Lambda\subset\bG$ be bounded, and $X$ be a maximal configuration coinciding with $\overline{X}$ on $\Lambda^c$.
    Let $[\Lambda]$ be the smallest $\sF$-measurable set such that $\Lambda$ is contained in $[\Lambda]\setminus\partial^{\mathrm{in}}[\Lambda]$. 
    Then,
    \begin{equation*}
        \begin{multlined}
            H_{\Lambda}(X_\Lambda \overline{X}_{\Lambda^c})-H_{\Lambda}(\overline{X}_\Lambda \overline{X}_{\Lambda^c})
            =-\sum_{x\in X_{[\Lambda]}}\mu_x+\sum_{x\in\overline{X}_{[\Lambda]}}\mu_x
            \\
            \ge -\sigma p_\ast\sum_{x\in X_{[\Lambda]}}v(x\mid X)+\sigma p_\ast\sum_{x\in\overline{X}_{[\Lambda]}}v(x\mid\overline{X})
            =0,
        \end{multlined}
    \end{equation*}
    where we used Assumption~\ref{asm:optimization} in the inequality and Lemma~\ref{lem:allocated volume screening} in the last equality.

    For the other direction, let $Y$ be a periodic ground state which is not perfect. 
    For every $x \in Y$, define
    \begin{equation*}
        g_Y(x):=\sigma p_\ast v(x\mid Y)-\mu_x,
    \end{equation*}
    which is non-negative by Assumption~\ref{itm:minorization}.
    Moreover, there exists $y_0 \in Y$ such that
    \begin{equation*}
        \delta:=g_Y(y_0)>0.
    \end{equation*}

    Fix a perfect configuration $\overline{X}$, and choose a supercell $G$
    common to $Y$ and $\overline{X}$ and larger than the hard-core interaction
    radius such that $y_0\in Y_G$.  
    We have
    \begin{equation}
        \label{eqn:strict gap in a period cell}
        \sum_{y\in Y_G}\mu_y
        = \sigma p_\ast v(Y_G\mid Y) - \sum_{y \in Y_G} g_Y(y)
        \le
        \sigma p_\ast v(Y_G\mid Y)-\delta.
    \end{equation}

    For $n\ge1$, set
    \begin{equation*}
        B_n:=\onion[n]{G},
        \qquad
        C_n:=\onion[n+{3}]{G},
        \qquad
        A_n:=C_n\setminus B_n.
    \end{equation*}
    Thus, $A_n$ is the annulus between the inner block $B_n$
    and the outer block $C_n$.  
    Choose a maximal configuration $Z_n$ that agrees with $\overline{X}$ on $B_n$ and with $Y$ on $C_n^c$ (this is always possible by Lemma~\ref{lem:maximal completion property}). 
    Since $B_n$ contains $(2n+1)^d$ copies of $G$, the finite-volume Hamiltonian of
    $Z_n$ in $C_n$, with boundary condition $Y_{C_n^c}$, is
    \begin{equation}
        \label{eqn:energy of interpolating configuration}
        \begin{aligned}
            H_{C_n}((Z_n)_{C_n} Y_{C_n^c})
            &=
            -(2n+1)^d\sum_{x\in\overline{X}_G}\mu_x
            -\sum_{z\in Z_n\cap A_n}\mu_z\\
            &=
            -(2n+1)^d\,\sigma p_\ast
            v(\overline{X}_G\mid\overline{X})
            -\sum_{z\in Z_n\cap A_n}\mu_z.
        \end{aligned}
    \end{equation}
    The second equality follows since $\overline X$ is a perfect configuration.

    The block $C_n$ contains $(2n+7)^d$ copies of $G$.  Periodicity of $Y$
    and~\eqref{eqn:strict gap in a period cell} give
    \begin{equation}
        \label{eqn:energy of nonperfect periodic configuration}
        \begin{aligned}
            H_{C_n}(Y_{C_n} Y_{C_n^c})
            &=-(2n+7)^d\sum_{y\in Y_G}\mu_y\\
            &\ge
            -(2n+7)^d
            \left[\sigma p_\ast v(Y_G\mid Y)-\delta\right].
        \end{aligned}
    \end{equation}

    Let
    \begin{equation*}
        \Delta_n
        :=
        H_{C_n}((Z_n)_{C_n} Y_{C_n^c})
        -
        H_{C_n}(Y_{C_n} Y_{C_n^c}).
    \end{equation*}
    Subtracting~\eqref{eqn:energy of nonperfect periodic configuration}
    from~\eqref{eqn:energy of interpolating configuration} and dividing by
    $(2n+7)^d$ yields
    \begin{equation}
        \label{eqn:bulk boundary energy comparison}
        \begin{aligned}
            \frac{\Delta_n}{(2n+7)^d}
            \le{}&
            -\frac{1}{(2n+7)^d}
            \sum_{z\in Z_n\cap A_n}\mu_z\\
            &-\frac{(2n+1)^d}{(2n+7)^d}\,
            \sigma p_\ast v(\overline{X}_G\mid\overline{X})
            +\sigma p_\ast v(Y_G\mid Y)-\delta.
        \end{aligned}
    \end{equation}

    The boundary annulus $A_n$ contains $(2n+7)^d-(2n+1)^d=O(n^{d-1})$ cells.  
    Since every cell contains only finitely many possible particle centers and the chemical potentials are bounded, the first term on the right-hand side of~\eqref{eqn:bulk boundary energy comparison} tends to zero as $n \to \infty$.
    Therefore, by Lemma~\ref{lem:common supercell equal volume},
    \begin{equation*}
        \limsup_{n\to\infty}
        \frac{\Delta_n}{(2n+7)^d}
        \le-\delta<0.
    \end{equation*}
    Hence $\Delta_n<0$ for every sufficiently large $n$.  But $Z_n$ agrees
    with $Y$ outside the bounded region $C_n$, so this is a finite
    modification that lowers the energy of $Y$, contradicting the assumption
    that $Y$ is a ground state.  Therefore $Y$ is perfect.
\end{proof}

\subsection{The Peierls condition}

In this subsection, we verify the Peierls condition under Assumptions~\ref{asm:optimization} and~\ref{asm:supercell}.

Our contour formalism follows~\cite{zahradnik1984alternate}.
Recall that $F=A\cdot[0,1)^d$ is a fixed supercell, defined above in Assumption~\ref{asm:supercell}.
We index its translated cells $K_n$ by $n \in \Z^d$ such that $K_n:=A\cdot n+F$.
We call two
distinct cells $K_n$ and $K_{n'}$ \emph{adjacent} if $\norm{n-n'}_\infty=1$.
Connectivity of sets of cells is understood with respect to this adjacency relation.  

Write the perfect configurations as
$\overline{X}^{\,1},\ldots,\overline{X}^{\,r}$.
Refining the earlier definition, we call a cell $K$  \textbf{$q$-correct} with respect to a maximal configuration $X$ if
    $X_{\onion[1]{K}}
    =
    \overline{X}^{\,q}_{\onion[1]{K}}$.
The cell is called \textbf{correct} if it is $q$-correct for some $q$, and \textbf{incorrect} otherwise. 

A contour $\gamma$ is a restriction of $X$ to a bounded connected component of incorrect cells (called its \textbf{support}), denoted by $\gamma:=\bigl(\supp{\gamma},X_{\supp{\gamma}}\bigr)$.
Note that $\supp{\gamma}^c$ has one unbounded connected component and  finitely many bounded connected components.  The unbounded component is called the
\textbf{exterior} of $\gamma$, denoted by $\ext{\gamma}$.  The union of the bounded components is called the
\textbf{interior}, denoted by $\intr{\gamma}$.  
The cells in $\supp{\gamma}^c$ that are adjacent
to $\supp{\gamma}$ are  $q$-correct for some $q$, and we say $q$ is the \textbf{phase label} of that component.
The label of the unbounded component is called the \textbf{type} of the
contour.  
For each $q$, let $\interior{q}{\gamma}$ be the union of the 
components having phase label $q$.  Then
\begin{equation}\label{eq:interior}
    \intr{\gamma}
    =
    \bigcup_{q=1}^r\interior{q}{\gamma}.
\end{equation}
Further, let 
 \begin{equation} \label{eq: V gamma} 
    V(\gamma)
    :=
    \supp{\gamma}\cup\intr{\gamma}.
\end{equation}

    We say that the \textbf{Peierls condition} holds if there exists a constant $\tau>0$ such that for every contour $\gamma=(\supp{\gamma},X_{\supp{\gamma}})$, 
    \begin{equation}
        \label{eqn:Peierls condition}
        H_{\supp{\gamma}}(X_{\supp{\gamma}} X_{\supp{\gamma}^c})-H_{\supp{\gamma}}(\overline{X}_{\supp{\gamma}} \overline{X}_{\supp{\gamma}^c})
        \ge \tau\abs{\supp{\gamma}},
    \end{equation}
    where $|\supp{\gamma}|$ is the number of cells contained in $\gamma$.

\begin{proposition}
    \label{prop:Peierls condition}
    Under Assumptions~\ref{asm:optimization} and~\ref{asm:supercell}, the Peierls condition holds.
\end{proposition}

\begin{remark}
    In applications, it may be possible to determine a value for $\tau$ explicitly.
    Otherwise, its existence is proved non-constructively below.
\end{remark}

The main ingredient in the proof of Proposition~\ref{prop:Peierls condition} is the following lemma, which guarantees a uniformly positive gap in the allocated volume within a uniformly bounded neighborhood of every incorrect cell.
Similar ideas first appeared in~\cite{holsztynski1978peierls}.

\begin{lemma}
    \label{lem:defect nearby}
    There exist constants $\delta>0$ and $n\ge 1$ such that for every cell $K\in\cT^F$ and maximal configuration $X\in\Omega$ with respect to which $K$ is incorrect, there exists $x\in X_{\onion[n]{K}}$ such that $\sigma p_\ast\, v(x\mid X)\ge\mu_x+\delta$.
\end{lemma}

\begin{proof}
    Without loss of generality, we may assume that $K=F$ (cf. Assumption~\ref{asm:supercell}). 
 Arguing toward a contradiction, assume that there exists a sequence of maximal configurations $(X_n)_{n\ge 1}$ such that for every $n\ge 1$, (a) $F$ is incorrect with respect to $X_n$, and (b) for every $x\in X_n\cap \onion[n]{F}$, $\sigma p_\ast\, v(x\mid X_n)<\mu_x+\frac{1}{n}$. (Here the intersection
 $X_n\cap \onion[n]{F}$ means the restriction of $X_n$ to (the sites) in $\onion[n]F$.)
    Note that $X_n\cap \onion[1]{F}\ne\emptyset$ for all $n\ge 1$: Otherwise, we can add a tile to $X_n$ in $F$, contradicting the assumption that $X_n$ is maximal. 
    By passing to a subsequence if necessary, we may further assume that there exists a maximal configuration $X\in\Omega$ such that $X_n$ converges to $X$.
    By Remark~\ref{rem:maximality is local}, $X$ is again maximal.
    Let $x\in X$.
    Then, for every sufficiently large $n$, $x\in X_n\cap \onion[n]{F}$ and thus by Lemma~\ref{lem:translation invariance and lower semicontinuity},
    \begin{equation*}
        \mu_x
        \le \sigma p_\ast\, v(x\mid X)
        =\lim_{n\to\infty}\sigma p_\ast\, v(x\mid X_n)
        \le \liminf_{n\to\infty}\Big(\mu_x+\frac{1}{n}\Big)
        =\mu_x.
    \end{equation*}
    Hence, $X$ is a perfect configuration.
    However, for every sufficiently large $n$, $X_n$ coincides with $X$ on $\onion[1]{F}$, and so $F$ is correct with respect to $X_n$, which is a contradiction.
\end{proof}

We also need the following identity for the total allocated volume in the support of a contour.

\begin{lemma} \label{lem:equivalence to lattice volume} 
    Let $S\subset\R^d$ be an $\bL$-invariant space with an $\bL$-invariant, locally finite measure $\lambda$
    and let $F$ be the fixed supercell.
    For every contour $\gamma=(\supp{\gamma},X_{\supp{\gamma}})$, 
    \begin{equation*}
        v(X_{\supp{\gamma}}\mid X)=\lambda(F\cap S)\abs{\supp{\gamma}}.
    \end{equation*}
\end{lemma}

\begin{proof}
    Let $\overline{X}$ be the perfect configuration such that $X_{\partial^{\mathrm{in}} V(\gamma)}=\overline{X}_{\partial^{\mathrm{in}} V(\gamma)}$.
    We compute
    \begin{equation*}
        \begin{multlined}
            v(X_{\supp{\gamma}}\mid X)
            =v(X_{V(\gamma)}\mid X)-v(X_{\intr{\gamma}}\mid X)
            \\
            =v(\overline{X}_{V(\gamma)}\mid\overline{X})-v(X_{\intr{\gamma}}\mid X)
            =\lambda(F\cap S)\abs{\supp{\gamma}},
        \end{multlined}
    \end{equation*}
    where we used Lemmas~\ref{lem:allocated volume screening} and~\ref{lem:common supercell equal volume} in the second and third equalities, respectively.
\end{proof}

\begin{proof}[Proof of Proposition~\ref{prop:Peierls condition}]
    By Lemma~\ref{lem:defect nearby}, there exist constants $\delta>0$ and $n\ge 1$ such that for every cell $K\in\cT^F$ contained in $\supp{\gamma}$, there exists $x\in X_{\onion[n]{K}}$ with $\sigma p_\ast\, v(x\mid X)\ge\mu_x+\delta$; in fact, $x\in X_{\supp{\gamma}}\cap \onion[n]{K}$ since every cell $K\in\cT^F$ contained in $\supp{\gamma}^c$ is correct with respect to $X$.

    Now, by Assumption~\ref{asm:optimization},
    \begin{equation}
        \label{eqn:lower bound on contour energy}
        H_{\supp{\gamma}}(X_{\supp{\gamma}} X_{\supp{\gamma}^c})
        =-\sum_{x\in X_{\supp{\gamma}}}\mu_x
        \ge -\sigma p_\ast\sum_{x\in X_{\supp{\gamma}}}v(x\mid X)+\frac{\delta}{(2n+1)^d}\abs{\supp{\gamma}},
    \end{equation}
    where the factor $1/(2n+1)^d$ is loosely interpreted as the minimum density of bad cells in a contour, making use of the fact that $\onion[n]{K}$ consists of $(2n+1)^d$ cells for any given cell $K$.  and
    \begin{equation}
        \label{eqn:identity for ground state contour energy}
        H_{\supp{\gamma}}(\overline{X}_{\supp{\gamma}} \overline{X}_{\supp{\gamma}^c})
        =-\sum_{x\in \overline{X}_{\supp{\gamma}}}\mu_x
        =-\sigma p_\ast\sum_{x\in \overline{X}_{\supp{\gamma}}}v(x\mid \overline{X}).
    \end{equation}
    Combining~\eqref{eqn:lower bound on contour energy},~\eqref{eqn:identity for ground state contour energy}, and Lemma~\ref{lem:equivalence to lattice volume}, we have that
    \begin{equation*}
        H_{\supp{\gamma}}(X_{\supp{\gamma}} X_{\supp{\gamma}^c})-H_{\supp{\gamma}}(\overline{X}_{\supp{\gamma}} \overline{X}_{\supp{\gamma}^c})
        \ge\frac{\delta}{(2n+1)^d}\abs{\supp{\gamma}},
    \end{equation*}
    so~\eqref{eqn:Peierls condition} holds with $\tau=\delta/(2n+1)^d$.
\end{proof}

\subsection{Screening and locality}

    The allocation in Example~\ref{ex:Z2 counting allocation} has the
    following locality property: if two maximal configurations containing
    $x$ agree at every lattice vertex within graph distance $2$ of $x$, then
    they give the same function $\phi_x$.  
    Assumption~\ref{asm:supercell} can be verified by choosing $F$ large enough that every
    lattice vertex within graph distance $2$ of a site in $F$ belongs to
    $\onion[1]{F}$.  

We now upgrade this observation to a general statement.
We say that a volume allocation $\set{\phi_x(\cdot\mid X)}$ is \textbf{local} at $X\in\Omega$ if, for every $x=(v,i)\in X$, there exists a locality radius $r=r(X,x)<\infty$ such that, for every $Y\in\Omega$ such that $Y_{B_{r}(v)}=X_{B_{r}(v)}$, $\phi_x(\cdot\mid Y)=\phi_x(\cdot\mid X)$ $\lambda$-a.e.

\begin{proposition}
    \label{prop: locality implies supercell assumption}
    If Assumption~\ref{asm:optimization} holds with a volume allocation which is local at every perfect configuration, then Assumption~\ref{asm:supercell} holds.
\end{proposition}

\begin{proof}
    By Assumption~\ref{asm:optimization}, the set $\mathcal P$ of perfect
    configurations is finite.  Fix $t\in\bL$ and
    $\overline X\in\mathcal P$.  Translation by $t$ preserves the hard-core
    constraint and maximality.  Moreover, $\mu_{t+x}=\mu_x$, and
    Lemma~\ref{lem:translation invariance and lower semicontinuity} gives 
     \begin{equation*}
        v(t+x\mid t+\overline X)=v(x\mid\overline X).
    \end{equation*}
    Hence $t+\overline X$ is also perfect and therefore belongs to $\mathcal P$.

    For each $\overline X\in\mathcal P$, define its period lattice by
    \begin{equation*}
        \operatorname{Per}(\overline X)
        :=\set{t\in\bL\mid t+\overline X=\overline X}.
    \end{equation*}
    All translates of $\overline X$ belong to the finite set $\mathcal P$, so
    $\operatorname{Per}(\overline X)$ has finite index in $\bL$.  Since
    $\mathcal P$ is finite,
    \begin{equation*}
        \bK:=\bigcap_{\overline X\in\mathcal P}
        \operatorname{Per}(\overline X)
    \end{equation*}
    is a finite-index sublattice of $\bL$ and a common period lattice for all
    perfect configurations.
    
 Let us introduce the parameter
    \begin{equation*}
        r_\phi:=\max_{\overline{X}\in\cP,x\in\overline{X}} r(X,x) <\infty,
    \end{equation*}
which is well defined by the locality assumption and translation invariance of the volume allocation.
Since there are finitely many different tiles and each tile is bounded, we can also define
    \begin{multline*}
        r_\hc:=\inf\set{r\ge 0\mid\text{if }(v,i),(u,j)\in\bG\times\set{1,\dots,k}\text{ and }\rmd(v,u)\ge r,\\ \text{ then }\intr{v+T_{i}}\cap\intr{u+T_{j}}=\emptyset}<\infty.
    \end{multline*}
   Set $r:=\max\set{r_\phi,r_\hc}$.
    
    Let $F$ be a (sufficiently large) supercell of $\bK$ such that 
    \begin{equation}
        \label{eqn: locality radius contained in one-enlargement}
        B_r(v)\subset\onion[1]{F}\text{ for all }v\in F.
    \end{equation}
    We verify Assumption~\ref{asm:supercell} for this choice of a supercell.
    Suppose that $F$ is $\overline{X}$-correct with respect to a maximal configuration $X$.
    Then, $X_{\onion[1]{F}}=\overline{X}_{\onion[1]{F}}$.
    After combining this with~\eqref{eqn: locality radius contained in one-enlargement}, locality yields~\eqref{eqn:allocation screening}.
\end{proof}
This proposition has a direct implication for nearest-neighbor models with discrete volume allocation.

\begin{corollary}
    \label{cor: nn model assumption implication}
    For nearest-neighbor models with the discrete volume allocation,
    Assumption~\ref{asm:optimization} implies
    Assumption~\ref{asm:supercell}.
\end{corollary}

\begin{proof}
    For a nearest-neighbor maximal model, by~\eqref{eqn:Z2 counting allocation}, the discrete volume allocation $\phi_x(\cdot\mid X)$ is completely determined by the restriction of a maximal configuration $X$ to within a graph distance of $2$ from $x\in X$.
    Hence, the discrete volume allocation is local at \textit{every} maximal configuration of the model, including the perfect ones, so Proposition~\ref{prop: locality implies supercell assumption} applies.
\end{proof}

The following example shows that the assumption of locality in Proposition~\ref{prop: locality implies supercell assumption} cannot be removed in general.

\begin{example}
    We give an example of a volume allocation that satisfies Assumption~\ref{asm:optimization} but does not satisfy Assumption~\ref{asm:supercell}.
    Consider the $\Z^2$ setting of Example~\ref{ex:Z2 counting allocation}, let $\lambda$ again be the counting measure, and assign to each occupied vertex the weight
    \begin{equation*}
        a_x(X):=1+\varepsilon\sum_{z\in X}2^{-\rmd_{\mathrm{gr}}(x,z)},
        \qquad \varepsilon>0.
    \end{equation*}
    Define the volume allocation for every $y \in \Z^2, x \in X$ by
    \begin{equation*}
        \widetilde{\phi}_x(y\mid X)
        := \begin{cases}
            1, & y=x, \\[1ex]
            \dfrac{a_x(X)}{\displaystyle\sum_{\substack{x'\in X\\x'\sim y}} a_{x'}(X)},
                & y\notin X,\; y\sim x, \\[2ex]
            0, & \text{otherwise}.
        \end{cases}
    \end{equation*}
    The resulting allocation satisfies Assumption~\ref{asm:optimization}.  
    Nevertheless, changing $X$ arbitrarily far from $F$ (without changing $X_{\onion[1]{F}}$) can change the weights of particles in $F$; thus, Assumption~\ref{asm:supercell} fails.
\end{example}

\section{Parameter estimates in Pirogov--Sinai theory} \label{sec:PS-theory}
In this section, our goal is to formulate an explicit list of sufficient conditions on the parameters, such as the inverse temperature $\beta$ and the Peierls constant $\tau$, for Pirogov--Sinai theory to apply to the maximal model.
Although Pirogov--Sinai arguments can be used to derive
inequalities of this kind~\cite{he2024high,he2026rigorous} (see
also~\cite{dereudre2025liquid}), the existing formulations do not provide the
resulting numerical bounds in a form directly applicable in our setting.
Obtaining such bounds requires carrying the constants through the full contour
and cluster-expansion calculation.

For this reason, we give a detailed treatment of the version of
Pirogov--Sinai theory~\cite{pirogov1975phase,pirogov1976phase} proposed by
Zahradn\'ik~\cite{zahradnik1984alternate}. 
Our exposition takes inspiration from~\cite[Chapter~7]{friedli2017statistical} and organizes the results in~\cite{zahradnik1984alternate} in a slightly different order to enhance the logical flow, though the proof of the main technical result, Theorem~\ref{thm:applyPS}, still mostly follows Zahradn\'ik's original.

Throughout, we work with the maximal model in the setting of Section~\ref{sec:the model}.
Via the supercell construction stated before Assumption~\ref{asm:supercell}, we may equivalently view the maximal model as a spin model on $\Z^d$ with a finite single-spin space and a translation-invariant, possibly infinite-valued interaction of range one (with respect to the sup-norm).
By Proposition~\ref{prop:perfect iff periodic ground state}, the periodic ground states of the model are the finitely many perfect configurations $\overline X^{\,1},\dots,\overline X^{\,r}$.
In addition, Proposition~\ref{prop:Peierls condition} supplies a constant $\tau>0$ such that every contour satisfies the Peierls condition~\eqref{eqn:Peierls condition}.

\subsection{Contour representation}

First, we recall the contour representation in the form adapted to the maximal model.

Let $X\in \Omega$ be a maximal configuration.
For a periodic ground state $\overline X^q$ and a bounded, $\sF$-measurable region $\Lambda$, let $\cC^q(\Lambda)$ be the set of contours of type $q$ whose support is contained in $\Lambda\setminus\partial^{\mathrm{in}}\Lambda$ (we say that these contours are \textbf{properly contained} in $\Lambda$).
Let $\cC^q:=\bigcup_{\Lambda\Subset \bG}\cC^q(\Lambda)$. 
Two contours $\gamma_1,\gamma_2\in\cC^q$ are compatible if their supports have disjoint closures.
Given a finite set $\Gamma\subset\cC^q$ of mutually compatible contours, a contour $\gamma\in\Gamma$ is external with respect to $\Gamma$ if its support is not contained in the interior of any other contour in $\Gamma$.

Henceforth, suppose that $\Lambda$ is a bounded, $\sF$-measurable region whose complement $\R^d\setminus\Lambda$ is connected. Define 
   \begin{equation}\label{eq: Zq}
\bZ^q(\Lambda):=\sum_{X\in \Omega_\Lambda^q} e^{-\beta H_{\Lambda}(X_{\Lambda}  \overline{X}^q_{\Lambda^c})},
  \end{equation}
where $H_{\Lambda}(X_{\Lambda}  \overline{X}^q_{\Lambda^c}) = -\sum_{x \in X_\Lambda} \mu_x$; cf.~\eqref{eq: periodic ground state}.
Let $\ext{\Gamma} = \cap_{\gamma \in \Gamma} \ext{\gamma}$.
Denote by $\Omega_{\Lambda}^q$ the set of configurations $X\in\Omega$ such that $X_{\Lambda^c}=\overline{X}^{\,q}_{\Lambda^c}$, and for each contour $\gamma$ of $X$, $V(\gamma) \subset \Lambda\setminus\partial^{\mathrm{in}} \Lambda$, where $V(\gamma)$ is defined in \eqref{eq: V gamma}. Let $X \in \Omega_{\Lambda}^q$ and let $\Gamma$ be the collection of (mutually compatible) external contours defined by $X$. We can rewrite
the Hamiltonian $H_{\Lambda}(X_{\Lambda}  X_{\Lambda^c})$ as follows: 
\begin{equation} \label{eq:rewrite-H}
    \begin{split}
        H_{\Lambda}(X_{\Lambda}  X_{\Lambda^c}) & = H_{\ext{\Gamma} \cap \Lambda}(X_{\ext{\Gamma} \cap \Lambda}  X_{\ext{\Gamma}^c \cup \Lambda^c})\\
        & \hspace*{1in}+ \sum_{\gamma \in \Gamma} \brackets[\Bigg]{H_{\supp{\gamma}}(X_{\supp{\gamma}}  X_{\supp{\gamma}^c}) + \sum_{p=1}^r H_{\interior{p}{\gamma}}(X_{\interior{p}{\gamma}}  X_{\interior{p}{\gamma}^c} ) } \\
        & = H_{\ext{\Gamma} \cap \Lambda}(\overline{X}^q_{\ext{\Gamma} \cap \Lambda}  \overline{X}^q_{\ext{\Gamma}^c \cup \Lambda^c})\\
        & \hspace*{1in}
        + \sum_{\gamma \in \Gamma} \brackets[\Bigg]{ H_{\supp{\gamma}}(X_{\supp{\gamma}}  X_{\supp{\gamma}^c}) + \sum_{p=1}^r H_{\interior{p}{\gamma}}(X_{\interior{p}{\gamma}}  X_{\interior{p}{\gamma}^c} ) },
    \end{split}
\end{equation}
where in the last equality we used the definition of $X$. 
In the sum on $X$ in \eqref{eq: Zq}, $X_{\interior{p}{\gamma}}$ ranges over $\Omega_{\interior{p}{\gamma}}^{p}$ as $X$ ranges over $\Omega_{\Lambda}^q$. 
Consequently,~\eqref{eq:rewrite-H} gives the following expression for the partition function:
\begin{equation*}
    \bZ^q(\Lambda) = \sum_{\substack{\Gamma\subset\cC^q(\Lambda)\\\text{compatible}\\\text{external}}} e^{-\beta H_{\ext{\Gamma} \cap \Lambda}(\overline{X}^q_{\ext{\Gamma} \cap \Lambda}  \overline{X}^q_{\ext{\Gamma}^c \cup \Lambda^c})} \prod_{\gamma \in \Gamma} \brackets[\Bigg]{ e^{-\beta H_{\supp{\gamma}}(X_{\supp{\gamma}}  X_{\supp{\gamma}^c})} \prod_{p=1}^r \bZ^p(\interior{p}{\gamma})}.
\end{equation*}
Now define
\begin{equation*}
    \begin{multlined}
        \Xi^q(\Lambda) := \frac{\bZ^q(\Lambda)}{e^{-\beta H_{\Lambda}(\overline{X}^q_\Lambda\overline{X}^q_{\Lambda^c})}}
        \\
        =\sum_{\substack{\Gamma\subset\cC^q(\Lambda)\\\text{compatible}\\\text{external}}}\prod_{\gamma\in\Gamma}
    \brackets[\Bigg]{
        \frac{e^{-\beta H_{\supp{\gamma}}(X_{\supp{\gamma}} X_{\supp{\gamma}^c})}}{e^{-\beta H_{\supp{\gamma}}(\overline{X}^q_{\supp{\gamma}} \overline{X}^q_{\supp{\gamma}^c})}}
        \prod_{p=1}^{r}\frac{\bZ^{p}(\interior{p}{\gamma})}{e^{-\beta H_{\interior{p}{\gamma}}(\overline{X}^q_{\interior{p}{\gamma}}\overline{X}^q_{\interior{p}{\gamma}^c})}}
    }.
    \end{multlined}
\end{equation*}
Introducing the weight
\begin{equation*}
    w^q(\gamma):=\frac{e^{-\beta H_{\supp{\gamma}}(X_{\supp{\gamma}} X_{\supp{\gamma}^c})}}{e^{-\beta H_{\supp{\gamma}}(\overline{X}^q_{\supp{\gamma}} \overline{X}^q_{\supp{\gamma}^c})}}
    \prod_{p}\frac{\bZ^{p}(\interior{p}{\gamma})}{\bZ^{q}(\interior{p}{\gamma})},
\end{equation*}
we rewrite $\Xi^q(\Lambda)$ as
\begin{equation}
    \label{eqn: external contour recursion}
    \Xi^q(\Lambda) =\sum_{\substack{\Gamma\subset\cC^q(\Lambda)\\\text{compatible}\\\text{external}}}\prod_{\gamma\in\Gamma}
    \brackets[\Bigg]{w^q(\gamma)\prod_{p}\frac{\bZ^q(\interior{p}{\gamma})}{e^{-\beta H_{\interior{p}{\gamma}}(\overline{X}^q_{\interior{p}{\gamma}}\overline{X}^q_{\interior{p}{\gamma}^c})}}}.
\end{equation}
We now apply~\eqref{eqn: external contour recursion} recursively to each
normalized interior partition function.
The iteration terminates when $\interior{p}{\gamma}$ is too small to \textit{properly} contain any other contour, in which case $\bZ^q(\interior{p}{\gamma})=e^{-\beta H_{\interior{p}{\gamma}}(\overline{X}^q_{\interior{p}{\gamma}}\overline{X}^q_{\interior{p}{\gamma}^c})}$.
Hence, we obtain the \textbf{contour representation} of the maximal model:
\begin{equation}
    \label{eqn: contour representation}
    \Xi^q(\Lambda) 
    =\sum_{\substack{\Gamma\subset\cC^q(\Lambda)\\\text{compatible}}}\prod_{\gamma\in\Gamma}w^q(\gamma).
\end{equation}

\subsection{Truncated contour model}
\label{sec: truncated contour model}

For a periodic ground state $\overline X^q$ and a contour $\gamma\in\cC^q$, define the truncated weight of $\gamma$ by
\begin{equation*}\label{eq:w-hat}
    \hat{w}^q(\gamma) = \min \set*{w^q(\gamma), e^{-\beta\tau\abs{\supp{\gamma}}/3}};
\end{equation*}
a contour $\gamma$ is called \textbf{stable} if $w^q(\gamma) \le e^{-\beta\tau\abs{\supp{\gamma}}/3}$, and \textbf{unstable} otherwise.
Define the truncated contour partition function
\begin{equation}
    \label{eqn: truncated contour model}
    \hat{\Xi}^q(\Lambda) := \sum_{\substack{\Gamma\subset\cC^q(\Lambda)\\\text{compatible}}} \prod_{\gamma \in \Gamma} \hat{w}^{q}(\gamma)
\end{equation}
and the corresponding pressure
\begin{equation}\label{eq:free_energy}
    s_{q} := \lim_{\Lambda \Uparrow \Z^d} \frac{1}{\beta\abs{\Lambda}} \ln \hat{\Xi}^q(\Lambda).
\end{equation}
Define the energy density of the periodic ground state $\overline X^q$ by
\begin{equation*}
    e_{q} = \lim_{\Lambda \Uparrow \Z^d} \frac{1}{\abs{\Lambda}} H_{\Lambda}(\overline{X}^{q}_\Lambda \overline{X}^{q}_{\Lambda^c}).
\end{equation*}
Finally, set 
\begin{equation}
    \label{eqn: metastable free energies}
    h_q := e_q - s_q,\quad
    h_0 := \min_{q} h_q,\quad\text{and}\quad
    a_q := h_q - h_0.
\end{equation}
We say that a periodic ground state $\overline X^q$ is \textbf{stable} if $a_q = 0$.
We stress that stability is a $\beta$-dependent property: \textit{a priori}, a periodic ground state stable for some value of $\beta$ may become unstable if this value changes.

Our overall goal is to establish crystallization, which will be shown in the next section, Theorem~\ref{thm:volume allocation crystallization}. This result is implied by the next theorem, which is
a slightly modified version of Zahradn\'ik's result in \cite[Theorem~1.7]{zahradnik1984alternate}.

\begin{theorem}[]\label{thm:applyPS}
There exists $\beta_0>0$ such that for all $\beta\ge\beta_0$ and all periodic ground states $\overline X^q$, the following statements hold. 
    \begin{enumerate}
        \item If $\gamma\in\cC^q$ is unstable, then
        \begin{equation}\label{eq:unstable-bound1}
              a_q \abs{\interior{}{\gamma}} \ge \frac{\tau}{3} \abs{\supp{\gamma}}.
        \end{equation}
        In particular, if the periodic ground state $\overline X^q$ is stable, then every contour of type $q$ is stable.
        \item For all bounded, $\sF$-measurable regions $\Lambda \subset \R^d$ with connected complement,
        \begin{align}
            \mathbf{Z}^{q}(\Lambda) {}&\ge e^{-\beta h_q\abs{\Lambda} - \abs{\exb \Lambda}},\label{eq:pslow} \\
            \mathbf{Z}^{q}(\Lambda) {}&\le e^{-\beta h_0 \abs{\Lambda} + 2 \abs{\exb \Lambda}} \label{eq:psup},
        \end{align}
    \end{enumerate}
where $\exb \Lambda:=\partial^{\mathrm{in}}\Lambda^c$.
\end{theorem}

To claim crystallization for the maximal hard-core model with chemical potentials $\vec\mu$, we would like
to find an explicit bound for $\beta_0$ in this theorem. The proofs in the literature do not pursue this
goal explicitly, and such bounds are not readily available. For this reason, we include a proof that the claims of Theorem~\ref{thm:applyPS} hold if $\beta_0$ satisfies the inequalities in Assumption~\ref{asm: parameters}, which we state next.
\begin{assumption}
    \label{asm: parameters}
    Let $\kappa$ denote the number of distinct restrictions of a maximal configuration to a fixed supercell, and recall the Peierls constant $\tau>0$.
    We assume that a constant $\beta_0>0$ has been chosen such that the following conditions are verified:
    \begin{equation}
        \label{eqn: parameters 1}
        \beta_0\tau\ge 9,
    \end{equation}
    \begin{equation}
        \label{eqn: parameters 2}
        \sum_{k=3^d}^\infty e^{-(\frac{1}{6}\beta_0\tau-1)k} \kappa^k (3^d-1)^{2k} \le 3^{-d},
    \end{equation}
    \begin{equation}
        \label{eqn: parameters 3}
        \frac{1}{\beta_0}\exp{-\frac{1}{6}\beta_0\tau \max \left \{ 3^d, \left(\frac{\tau}{3a}\right)^{d-1} \right \}} \le a \quad \text{for all } a>0.
    \end{equation}
\end{assumption}

\begin{remark}
    Condition~\eqref{eqn: parameters 1} will be used twice, in~\eqref{eqn: parameter condition usage 1} and~\eqref{eqn: parameter condition usage 2};
    Condition~\eqref{eqn: parameters 2} used once, in~\eqref{eqn: cluster expansion closure}, and Condition~\eqref{eqn: parameters 3} used once, in~\eqref{eqn: convoluted assumption}.
\end{remark}

The rest of this section is devoted to the proof of Theorem~\ref{thm:applyPS}.

\subsection{Cluster expansion}
Cluster expansions and their convergence are presented in many sources; for instance, see \cite[Chapter~5]{friedli2017statistical}.
Here, we will rely on the criterion proved by Ueltschi in \cite{ueltschi2004cluster}. We first cite it and then state its implications for the contour models used below.

\subsubsection{Ueltschi's criterion}

Let $\cC$ be a finite or countable set whose elements are called \textbf{polymers}.
The weight of a polymer is specified by a function $w:\cC\to\R$.
Let $\sim$ be an \textit{irreflexive} and \textit{symmetric} relation on $\cC$ called \textbf{compatibility}.
The (formal) \textbf{polymer partition function} is defined as
\begin{equation*}
    \Xi:=\sum_{\Gamma\subset\cC}\brackets[\Bigg]{\prod_{\gamma\in \Gamma}w(\gamma)}\brackets[\Bigg]{\prod_{\set*{\gamma,\gamma'}\subset \Gamma}\indicator{\gamma\sim \gamma'}}.
\end{equation*}

A finite multi-subset $M$ of $\cC$ is a \textbf{cluster} if it cannot be partitioned into two nonempty sub-(multi)sets $M_1,M_2$ such that every polymer in $M_1$ is compatible with every polymer in $M_2$.
Denote by $\cM$ the collection of all clusters.
For $n\ge 1$, let $G_n$ denote the complete graph on $n$ vertices, and write $G\subset G_n$ if $G$ is a spanning subgraph of $G_n$.
Define the function $\varphi:\cM\to\R$ by
\begin{equation*}
    \varphi(M):=\prod_{\gamma\in\cC}\frac{1}{n_\gamma(M)!}\sum_{\substack{G\subset G_n\\\text{connected}}}\prod_{\set*{i,j}\in E(G)}(-\indicator{\gamma_i\not\sim \gamma_j})
\end{equation*}
for $M=\mset*{\gamma_1,\dots,\gamma_n}$, where $n_\gamma(M)$ denotes the multiplicity of $\gamma$ in $M$.

\begin{theorem}[{Cluster expansion,~\cite{ueltschi2004cluster}}]
    \label{thm: Ueltschi criterion}
    Suppose that there exist functions $a,b:\cC\to[0,\infty)$ such that
    \begin{equation}
        \label{eqn: Ueltschi criterion}
        \sum_{\gamma\not\sim \gamma_\ast}e^{a(\gamma)+b(\gamma)}\abs{w(\gamma)}\le a(\gamma_\ast)\quad\text{for every }\gamma_\ast\in\cC.
    \end{equation}
    Then, for all $\gamma_1\in\cC$,
    \begin{equation}
        \label{eqn: cluster bound with one anchor}
        \sum_{M\ni \gamma_1}\abs{\varphi(M)}n_{\gamma_1}(M)\prod_{\gamma\in M\setminus\mset*{\gamma_1}}\left[e^{b(\gamma)}\abs{w(\gamma)}\right]\le e^{a(\gamma_1)}.
    \end{equation}
    If, in addition to~\eqref{eqn: Ueltschi criterion}, it holds that
    \begin{equation*}
        \sum_{\gamma\in\cC}e^{a(\gamma)}\abs{w(\gamma)}<\infty,
    \end{equation*}
    then $\Xi$ is well-defined, $\Xi>0$, and its logarithm admits the absolutely convergent series expansion
    \begin{equation*}
        \ln\Xi=\sum_{M\in\cM}\varphi(M)\prod_{\gamma\in M}w(\gamma).
    \end{equation*}
\end{theorem}

We note that Ueltschi~\cite{ueltschi2004cluster} originally defined clusters as ordered tuples rather than multisets, but this formulation of his criterion, taken from~\cite{he2026rigorous} and inspired by~\cite[Section~5.4]{friedli2017statistical}, is equivalent to the original and slightly lightens the notation.

\subsubsection{Application to contour models}

Below, contours of a fixed type $q$ play the role of polymers, and we take $\cC=\cC^q$. The collection of all clusters will be denoted $\cM^q$ below to emphasize that the contours forming the clusters are of type $q$.
\begin{theorem} \label{thm:FV7.29}
    Under Assumption~\ref{asm: parameters}, let $\beta\ge\beta_0$, let $\overline X^q$ be a periodic ground state, and let $\widetilde{w}^q:\cC^q\to\R$ be a translation-invariant function such that
    \begin{equation}\label{eq:stable_w}
        \abs{\widetilde{w}^q(\gamma)}
        \le e^{-\beta\tau\abs{\supp{\gamma}}/3},\quad\text{for all }\gamma\in\cC^q.
    \end{equation}
    Let $\widetilde{\Xi}^q(\Lambda)$ be the contour partition function obtained from~\eqref{eqn: contour representation} by replacing $w^q$ with $\widetilde{w}^q$.
    Then, the condition~\eqref{eqn: Ueltschi criterion} is verified with $\cC=\cC^q$ and
    \begin{equation*}
        a(\gamma):=\abs{\supp{\gamma}},\quad
        b(\gamma):=\beta\tau\abs{\supp{\gamma}}/6,\quad
        \text{for }\gamma\in\cC^q.
    \end{equation*}
    In particular, the pressure
    \begin{equation}
        \widetilde{s}_{q} := \lim_{\Lambda \Uparrow \Z^d} \frac{1}{\beta\abs{\Lambda}} \ln \widetilde{\Xi}^q(\Lambda)
    \end{equation}
    exists and is given by the absolutely convergent series 
    \begin{equation*}\label{eq:s_series}
        \widetilde{s}_q = \frac{1}{\beta}\sum_{\substack{M\in\cM^q \\\supp{M}\ni 0}} \frac{\varphi(M)}{\abs{\supp{M}}}\prod_{\gamma\in M}\widetilde{w}^q(\gamma),
    \end{equation*}
 where $\cM^q$ is the collection of all clusters of $\cC^q$, and $\supp{M} := \bigcup_{\gamma\in M} \supp{\gamma}$. 
    In addition, for all bounded, $\sF$-measurable regions $\Lambda$,
    \begin{equation}
        \label{eqn: bound on the finite-volume correction to the pressure}
        \abs{\frac{1}{\beta\abs{\Lambda}}\ln\widetilde{\Xi}^q(\Lambda)-\widetilde{s}_q}
        \le\frac{\abs{\partial^{\mathrm{in}} \Lambda}}{\beta 3^d \abs{\Lambda}}.
    \end{equation}
\end{theorem}

\begin{proof}
    We first verify Ueltschi's criterion.
    Let $\gamma_\ast\in\cC^q$.
    Then,
    \begin{equation*}
        \sum_{\substack{\gamma\in\cC^q \\ \gamma\not\sim\gamma_\ast}}e^{a(\gamma)+b(\gamma)}\abs{\widetilde{w}^q(\gamma)}
        \le 3^d\abs{\supp{\gamma_\ast}}\sum_{\substack{\gamma\in\cC^q\\\supp{\gamma}\ni 0}} e^{-(\frac{1}{6}\beta\tau-1)\abs{\supp{\gamma}}}.
    \end{equation*}
    We note that the minimum number of supercells in a contour is $3^d$ and, by~\cite[(3.39)]{friedli2017statistical}, the number of contours consisting of $k$ supercells and containing $0$ in their supports is bounded by $(3^d-1)^{2k}$.
    Hence,
    \begin{equation}
        \label{eqn: cluster expansion closure}
        \sum_{\substack{\gamma\in\cC^q\\\supp{\gamma}\ni 0}} e^{-(\frac{1}{6}\beta\tau-1)\abs{\supp{\gamma}}}
        \le \sum_{k=3^d}^\infty e^{-(\frac{1}{6}\beta\tau-1)k} \kappa^k (3^d-1)^{2k}
        \le 3^{-d},
    \end{equation}
    where the last inequality follows from~\eqref{eqn: parameters 2} since $\beta\ge\beta_0$, so the condition~\eqref{eqn: Ueltschi criterion} is verified. 

    Let $\Lambda$ be a bounded, $\sF$-measurable region.
    Then, using Theorem~\ref{thm: Ueltschi criterion} and the identity
    \begin{equation*}
        \frac{1}{|\supp{M}|} \sum_{i \in \Lambda} \indicator{i \in \supp{M}} = 1,
        \quad
        \text{for }\supp{M} \subset \Lambda,
    \end{equation*}
    we have
    \begin{equation*}
        \begin{multlined}
            \ln \widetilde{\Xi}^q(\Lambda) 
            = \sum_{v\in\Lambda\cap\cL^F}\sum_{\substack{M\in\cM^q(\Lambda)\\\supp{M}\ni v}} \frac{\varphi(M)}{\abs{\supp{M}}}\prod_{\gamma\in M}\widetilde{w}^q(\gamma)
            \\
            = \sum_{v\in\Lambda\cap\cL^F} \sum_{\substack{M\in\cM^q\\\supp{M}\ni v}} \frac{\varphi(M)}{\abs{\supp{M}}}\prod_{\gamma\in M}\widetilde{w}^q(\gamma)
            -\sum_{v\in\Lambda\cap\cL^F}\sum_{\substack{M\in\cM^q\setminus\cM^q(\Lambda)\\\supp{M}\ni v}} \frac{\varphi(M)}{\abs{\supp{M}}}\prod_{\gamma\in M}\widetilde{w}^q(\gamma).
        \end{multlined}
    \end{equation*}
    The second term is a boundary error: observe that
    \begin{equation*}
        \sum_{v\in\Lambda\cap\cL^F}\sum_{\substack{M\in\cM^q\setminus\cM^q(\Lambda)\\\supp{M}\ni v}} \frac{\abs{\varphi(M)}}{\abs{\supp{M}}}\prod_{\gamma\in M}\abs{\widetilde{w}^q(\gamma)}
        \le \sum_{v\in\partial^{\mathrm{in}}\Lambda\cap\cL^F}\sum_{\substack{M\in\cM^q\\\supp{M}\ni v}} \abs{\varphi(M)}\prod_{\gamma\in M}\abs{\widetilde{w}^q(\gamma)},
    \end{equation*}
    where, by~\eqref{eqn: cluster bound with one anchor} and~\eqref{eqn: cluster expansion closure}, the RHS is further bounded by
    \begin{equation*}
        \begin{multlined}
            \sum_{v\in\partial^{\mathrm{in}}\Lambda\cap\cL^F}\sum_{\substack{\gamma\in\cC^q\\\supp{\gamma}\ni v}} \abs{\widetilde{w}^q(\gamma)}
            \sum_{\substack{M\in\cM^q\\M\ni\gamma}} \abs{\varphi(M)} n_\gamma(M) \prod_{\gamma'\in M\setminus\mset{\gamma}}\abs{\widetilde{w}^q(\gamma')}
            \\
            \le \sum_{v\in\partial^{\mathrm{in}}\Lambda\cap\cL^F}\sum_{\substack{\gamma\in\cC^q\\\supp{\gamma}\ni v}} \abs{\widetilde{w}^q(\gamma)} e^{a(\gamma)}
            \le 3^{-d}\abs{\partial^{\mathrm{in}}\Lambda}.
        \end{multlined}
    \end{equation*}
    On the other hand, the quantity
    \begin{equation*}
        \sum_{\substack{M\in\cM^q\\\supp{M}\ni v}} \frac{\varphi(M)}{\abs{\supp{M}}}\prod_{\gamma\in M}\widetilde{w}^q(\gamma)
    \end{equation*}
    is independent of $v\in\bL$ and $\Lambda$ and yields an expression for $\widetilde{s}_q$ up to the constant factor $\beta^{-1}$.
    The theorem follows.
\end{proof}

\begin{corollary}[{\cite[Lemma~7.31]{friedli2017statistical}}]\label{thm:FV7.31}
    In the setting of Theorem~\ref{thm:FV7.29}, for all $L\ge 3^d$, it holds that
    \begin{equation*}
        \sum_{\substack{M\in\cM^q\\\supp{M}\ni 0\\\abs{\supp{M}}\ge L}}\abs{\varphi(M)}\prod_{\gamma\in M}\abs{\widetilde{w}^q(\gamma)}
        \le e^{-\beta\tau L/6}.
    \end{equation*}
\end{corollary}

\begin{proof}
    Write
    \begin{equation*}
        \sum_{\substack{M\in\cM^q\\\supp{M}\ni 0\\\abs{\supp{M}}\ge L}}\abs{\varphi(M)}\prod_{\gamma\in M}\abs{\widetilde{w}^q(\gamma)}
        \le \sum_{\substack{\gamma_\ast\in\cC^q\\\supp{\gamma_\ast}\ni 0}} \abs{\widetilde{w}^q(\gamma_\ast)}
        \sum_{\substack{M\in\cM^q\\M\ni\gamma_\ast}} \indicator{\abs{\supp{M}}\ge L}\abs{\varphi(M)}\prod_{\gamma\in M\setminus\mset{\gamma_\ast}}\abs{\widetilde{w}^q(\gamma)}.
    \end{equation*}
    Inserting
    \begin{equation*}
        \indicator{\abs{\supp{M}}\ge L} 
        \le e^{-\beta\tau L/6}e^{b(\gamma_\ast)} \prod_{\gamma\in M\setminus\mset{\gamma_\ast}} e^{b(\gamma)}
    \end{equation*}
    into the above and using~\eqref{eqn: cluster bound with one anchor}, we bound the LHS by
    \begin{equation*}
        \begin{multlined}
            e^{-\beta\tau L/6} \sum_{\substack{\gamma_\ast\in\cC^q\\\supp{\gamma_\ast}\ni 0}} \abs{\widetilde{w}^q(\gamma_\ast)} e^{b(\gamma_\ast)}
            \sum_{\substack{M\in\cM^q\\M\ni\gamma_\ast}} \abs{\varphi(M)}\prod_{\gamma\in M\setminus\mset{\gamma_\ast}} e^{b(\gamma)}\abs{\widetilde{w}^q(\gamma)}
            \\
            \le e^{-\beta\tau L/6} \sum_{\substack{\gamma_\ast\in\cC^q\\\supp{\gamma_\ast}\ni 0}} e^{a(\gamma_\ast)+b(\gamma_\ast)}\abs{\widetilde{w}^q(\gamma_\ast)}
            \le e^{-\beta\tau L/6},
        \end{multlined}
    \end{equation*}
    as required.
\end{proof}

\subsection{Proof of Theorem~\ref{thm:applyPS}}

    For~\eqref{eq:pslow}, we have that
    \begin{align}
        \mathbf{Z}^q(\Lambda)
        {}&= e^{-\beta e_q\abs{\Lambda}} \Xi^q(\Lambda)
        {}&&
        \nonumber
        \\
        {}&\ge e^{-\beta e_{q}\abs{\Lambda}} \hat{\Xi}^{q}(\Lambda)
        {}&&
        \nonumber
        \\
        {}&\ge e^{-\beta e_{q}\abs{\Lambda}+\beta s_q\abs{\Lambda}-3^{-d}\abs{\partial^{\mathrm{in}}\Lambda}}
        {}&&\text{by Theorem~\ref{thm:FV7.29} with }\widetilde{w}^q\equiv\hat{w}^q
        \nonumber
        \\
        {}&\ge e^{-\beta h_{q}\abs{\Lambda}-\abs{\exb\Lambda}}
        {}&&\text{since $\abs{\partial^{\mathrm{in}}\Lambda}\le 3^d\abs{\exb\Lambda}$}.
    \end{align} 

    Define the \textbf{level} of a contour $\gamma$ as the maximum integer $n\ge 1$ such that there exists a sequence of contours $(\gamma_1=\gamma, \gamma_2, \dots, \gamma_n)$ such that $\supp{\gamma_{i+1}} \subset \interior{}{\gamma_i}$ for all $1\le i\le n-1$.
    Define the \textbf{level} of a bounded, $\sF$-measurable region $\Lambda$ as the maximum level of a contour $\gamma$ such that $\supp{\gamma} \cup \interior{}{\gamma} \subset \Lambda$; if no such contour exists, the level of $\Lambda$ is set to zero.
    The proofs of~\eqref{eq:unstable-bound1} and~\eqref{eq:psup} are by induction on the level of $\gamma$ and $\Lambda$. 

    For the base case, we check only~\eqref{eq:psup} for regions $\Lambda$ of level $0$.
    Since $\cC^q(\Lambda)=\emptyset$, the truncated contour model in $\Lambda$ coincides with the original one, so
    \begin{align}
        \bZ^q(\Lambda){}&=e^{-\beta e_q\abs{\Lambda}} \hat{\Xi}^q(\Lambda)
        {}&&
        \nonumber
        \\
        {}&\le e^{-\beta e_{q}\abs{\Lambda}+\beta s_q\abs{\Lambda}+3^{-d}\abs{\partial^{\mathrm{in}}\Lambda}}
        {}&&\text{by Theorem~\ref{thm:FV7.29} with }\widetilde{w}^q\equiv\hat{w}^q
        \nonumber
        \\
        {}&\le e^{-\beta h_{q}\abs{\Lambda}+\abs{\exb\Lambda}}
        {}&&\text{since $\abs{\partial^{\mathrm{in}}\Lambda}\le 3^d\abs{\exb\Lambda}$}.
        \label{eqn: proof of upper bound in the base case}
    \end{align}
    
    Suppose hereafter that~\eqref{eq:unstable-bound1} and~\eqref{eq:psup} hold for all contours and regions of level at most $n\ge 0$.

    Let $\gamma$ be an unstable contour of level $n+1$.
    We have that
    \begin{align}
        {}& e^{-\beta\tau\abs{\supp{\gamma}}/3}
        {}&&
        \nonumber
        \\
        \le {}& w^q(\gamma) = \frac{e^{-\beta H_{\supp{\gamma}}(X_{\supp{\gamma}} X_{\supp{\gamma}^c})}}{e^{-\beta H_{\supp{\gamma}}(\overline{X}^q_{\supp{\gamma}} \overline{X}^q_{\supp{\gamma}^c})}}\prod_{p} \frac{\bZ^{p}(\interior{p}{\gamma})}{\bZ^{q}(\interior{p}{\gamma})}  
        {}&&\text{since $\gamma$ is unstable}
        \nonumber
        \\
        \le {}& e^{-\beta\tau\abs{\supp{\gamma}}} \prod_{p} \frac{\exp{-\beta h_0 \abs{\interior{p}{\gamma}} + 2\abs{\exb \interior{p}{\gamma}}}}{\exp{-\beta h_q \abs{\interior{p}{\gamma}} - \abs{\exb \interior{p}{\gamma}}}}
        {}&&
            \end{align}
where the last line follows by the Peierls condition~\eqref{eqn:Peierls condition},~\eqref{eq:pslow},
 and the induction hypothesis for~\eqref{eq:psup}.
        
    Thus,
    \begin{equation*}
        \begin{multlined}
            -\frac{\beta\tau}{3} \abs{\supp{\gamma}} 
            \le -\beta\tau\abs{\supp{\gamma}} + \sum_{p} \left( \beta a_q |\interior{p}{\gamma}| + 3|\exb\interior{p}{\gamma}| \right) 
            \\
            \le -\beta\tau \abs{\supp{\gamma}} + \beta a_q |\interior{}{\gamma}| + 3 \abs{\supp{\gamma}},
        \end{multlined}
    \end{equation*}
    which completes the inductive step for~\eqref{eq:unstable-bound1} provided that
    \begin{equation}
        \label{eqn: parameter condition usage 1}
        \beta\tau\ge 9.
    \end{equation} 

    Before proceeding to the inductive step for~\eqref{eq:psup}, we need some terminology.
    We say that a contour $\gamma\in\cC^q$ is \textbf{small} if it is itself stable and so is every $\gamma'\in\cC^q$ such that $\supp{\gamma'} \subset \interior{}{\gamma}$, and \textbf{large} otherwise.
    In the following, we use the subscript ${}_\text{small}$ to denote analogues of previously defined quantities where only small contours contribute, e.g., $\Xi^q_{\text{small}}(\cdot)$ and $\bZ^q_{\text{small}}(\cdot)$ (via~\eqref{eqn: contour representation}).

    Now, let $\Lambda$ be a region of level $n+1$.
    By isolating the \textit{large} external contours, we get
    \begin{equation}
        \label{eq:ZLambda}
        \mathbf{Z}^{q}(\Lambda) 
        = \sum_{\substack{\Gamma\subset\cC^q(\Lambda) \\ \text{C.E.L.}}} \mathbf{Z}_{\text{small}}^{q}(\ext{\Gamma}\cap\Lambda) 
        \prod_{\gamma \in \Gamma} 
        \brackets[\Bigg]{
            e^{-\beta H_{\supp{\gamma}}(X_{\supp{\gamma}} X_{\supp{\gamma}^c})} \prod_{p} \mathbf{Z}^{p}(\interior{p}{\gamma})
        },
    \end{equation}
    where $\ext{\Gamma} := \bigcap_{\gamma \in \Gamma} \ext{\gamma}$ and C.E.L. stands for ``compatible, external, and large".
    By the definition of smallness, each set of compatible contours that contributes to $\Xi^q_{\text{small}}(\cdot)$ consists of \textit{stable} contours.
    Thus, arguing as in~\eqref{eqn: proof of upper bound in the base case} using Theorem~\ref{thm:FV7.29} with $\widetilde{w}^q\equiv w^q\indicator{\text{small}}$, we get that
    \begin{equation}\label{eq:Zstable}
        \mathbf{Z}_{\text{small}}^{q}(\ext{\Gamma}\cap\Lambda) 
        = e^{-\beta e_{q}\abs{\ext{\Gamma}\cap\Lambda}}
        \Xi_{\text{small}}^{q}(\ext{\Gamma}\cap\Lambda)  
        \le e^{-\beta h_{q,\text{small}}\abs{\ext{\Gamma}\cap\Lambda}+\abs{\exb (\ext{\Gamma}\cap\Lambda)}}.
    \end{equation}
    On the other hand, applying the induction hypothesis for~\eqref{eq:psup} to $\mathbf{Z}^{p}(\interior{p}{\gamma})$, we have that 
    \begin{equation}\label{eq:Zprod}
        \prod_{p} \mathbf{Z}^{p}(\interior{p}{\gamma}) 
        \le \prod_{p} e^{-\beta h_0 |\interior{p}{\gamma}| + 2 |\exb \interior{p}{\gamma}|} 
        \le e^{-\beta h_0 \abs{\interior{}{\gamma}}+2\abs{\supp{\gamma}}},
    \end{equation}
    since $\exb \interior{p}{\gamma}\cap\exb \interior{p'}{\gamma}=\emptyset$ for $p\ne p'$.
    Inserting~\eqref{eq:Zstable} and~\eqref{eq:Zprod} into~\eqref{eq:ZLambda} yields
    \begin{equation*}
        \begin{multlined}
            \mathbf{Z}^{q}(\Lambda) 
            \le \sum_{\substack{\Gamma\subset\cC^q(\Lambda) \\ \text{C.E.L.}}} 
            e^{-\beta h_{q,\text{small}}\abs{\ext{\Gamma}\cap\Lambda}+\abs{\exb (\ext{\Gamma}\cap\Lambda)}} 
            \\
            \times
            \prod_{\gamma \in \Gamma} 
            \brackets[\bigg]{
                e^{-\beta H_{\supp{\gamma}}(X_{\supp{\gamma}} X_{\supp{\gamma}^c})-\beta h_0 \abs{\interior{}{\gamma}}+2\abs{\supp{\gamma}}}
            }.
        \end{multlined}
    \end{equation*}
    We bound
    \begin{equation*}
        \abs{\exb (\ext{\Gamma}\cap\Lambda)}+2\sum_{\gamma\in\Gamma}\abs{\supp{\gamma}}
        \le \abs{\exb\Lambda}+3\sum_{\gamma\in\Gamma}\abs{\supp{\gamma}},
    \end{equation*}
    and, using the Peierls condition~\eqref{eqn:Peierls condition},
    \begin{equation*}
        \begin{multlined}
            -h_{q,\text{small}}\abs{\ext{\Gamma}\cap\Lambda}-\sum_{\gamma\in\Gamma}\brackets[\Big]{H_{\supp{\gamma}}(X_{\supp{\gamma}} X_{\supp{\gamma}^c})+h_0 \abs{\interior{}{\gamma}}}
            \\
            \le -h_0\abs{\Lambda}-(h_{q,\text{small}}-h_0)\abs{\ext{\Gamma}\cap\Lambda}-(\tau+e_q-h_0)\abs{\supp{\Gamma}},
        \end{multlined}
    \end{equation*}
    at which point we note that
    \begin{equation*}
        h_{q,\text{small}}-h_0
        =(h_{q,\text{small}}-h_q)+(h_q-h_0)
        =(s_q-s_{q,\text{small}})+a_q
        \ge a_q,
    \end{equation*}
    and 
    \begin{equation*}
        e_q-h_0=(e_q-h_q)+(h_q-h_0)=s_q+a_q\ge 0.
    \end{equation*}
    Combining the above estimates, we conclude that
    \begin{equation}\label{eq:Z_bd}
        \bZ^q(\Lambda)
        \le e^{-\beta h_0\abs{\Lambda}+\abs{\exb\Lambda}} \sum_{\substack{\Gamma\subset\cC^q(\Lambda)\\ \text{C.E.L.} }} 
        e^{-\beta a_{q}\abs{\ext{\Gamma}\cap\Lambda}} 
        \prod_{\gamma \in \Gamma} e^{-(\beta\tau-3)\abs{\supp{\gamma}}}.
    \end{equation}
    It remains to prove that the sum in~\eqref{eq:Z_bd} is bounded above by $e^{\abs{\exb\Lambda}}$.
    We rely on the following claim.

    \begin{claim}
        The following statements hold.
        \begin{enumerate}[label=(\alph*),ref=\theclaim(\alph*)]
            \item If $a_q=0$, then no contour $\gamma\in\cC^q$ of level at most $n+1$ is large.
            \item If $a_q>0$, then every large contour $\gamma\in\cC^q$ of level at most $n+1$ is geometrically large:
            \begin{equation}
                \label{eqn: large implies geometrically large}
                \abs{\supp{\gamma}} \ge \max\set[\bigg]{3^d,\bigg(\frac{\tau}{3a_{q}}\bigg)^{d-1}}.
            \end{equation}
            \label{itm: large implies geometrically large}
        \end{enumerate}
    \end{claim}

    \begin{proof}
        At this point in the induction, we know that~\eqref{eq:unstable-bound1} holds whenever $\gamma\in\cC^q$ is unstable and of level at most $n+1$.
        If $a_q=0$, then~\eqref{eq:unstable-bound1} implies that no contour $\gamma\in\cC^q$ of level at most $n+1$ is unstable, which implies the first statement by the definition of largeness.

        Suppose now that $a_q>0$.
        Moreover, suppose that $\gamma\in\cC^q$ is a large contour of level at most $n+1$.
        By definition, either $\gamma$ is unstable, or there exists an unstable contour $\gamma'\in\cC^q$, of level at most $n$, such that $\supp{\gamma'}\subset\interior{}{\gamma}$.
        In the first sub-case, we bound
        \begin{align}
            \abs{\supp{\gamma}}
            {}&\ge \abs{\exb \interior{}{\gamma}} 
            \nonumber
            \\
            {}&\ge \abs{\interior{}{\gamma}}^{\frac{d-1}{d}} 
            {}&&\text{by the isoperimetric inequality}
            \nonumber
            \\
            {}&\ge \left(\frac{\tau}{3 a_{q}}\abs{\supp{\gamma}}\right)^{\frac{d-1}{d}}
            {}&&\text{by~\eqref{eq:unstable-bound1}},
            \label{eqn: geometrically large first sub-case}
        \end{align}
        which implies~\eqref{eqn: large implies geometrically large}.
        In the second sub-case, combining the estimate
        \begin{align}
            a_q \abs{\interior{}{\gamma'}} 
            {}&\ge \frac{\tau}{3} \abs{\supp{\gamma'}}
            {}&&\text{by~\eqref{eq:unstable-bound1}}
            \nonumber
            \\
            {}&\ge \frac{\tau}{3} \abs{\interior{}{\gamma'}}^{\frac{d-1}{d}}
            {}&&\text{by applying~\eqref{eqn: geometrically large first sub-case} to $\gamma'$}
        \end{align}
        with the geometric fact $\abs{\interior{}{\gamma}}\ge\abs{\interior{}{\gamma'}}$, and using the isoperimetric inequality again, we get
        \begin{equation*}
            \abs{\supp{\gamma}}
            \ge \abs{\interior{}{\gamma}}^{\frac{d-1}{d}} 
            \ge \abs{\interior{}{\gamma'}}^{\frac{d-1}{d}} 
            \ge \bigg(\frac{\tau}{3a_{q}}\bigg)^{d-1},
        \end{equation*}
        as required.
    \end{proof}

    By the claim, if $a_q=0$, then only the empty collection $\Gamma=\emptyset$ contributes to the sum in~\eqref{eq:Z_bd}, the sum reduces to one, and the proof is complete.
    Otherwise, $a_q>0$.
    We introduce an auxiliary model on the same set of contours $\cC^q$ but with weight
    \begin{equation*}
        w^q_\star(\gamma):=
        \begin{cases}
            e^{-(\beta\tau-5)\abs{\supp{\gamma}}} & \text{if }\abs{\supp{\gamma}}\ge \max\set[\big]{3^d,\big(\frac{\tau}{3a_{q}}\big)^{d-1}}\\
            0 & \text{otherwise}
        \end{cases}
        ,
    \end{equation*}
    and use the subscript ${}_\star$ to denote associated quantities such as the contour partition function $\Xi^q_\star(\cdot)$ and the pressure $s_{q,\star}$.
    Provided that
    \begin{equation}
        \label{eqn: parameter condition usage 2}
        \beta\tau\ge 15/2,
    \end{equation}
    we have that $\beta\tau-5\ge\beta\tau/3$, so Theorem~\ref{thm:FV7.29} and Corollary~\ref{thm:FV7.31} apply, with $\widetilde{w}^q\equiv w^q_\star$, to the auxiliary model.
    In particular, by Corollary~\ref{thm:FV7.31},
    \begin{equation}
        \label{eqn: bound on the auxiliary pressure}
        \abs{s_{q,\star}}
        \le \frac{1}{\beta}\sum_{\substack{M\in\cM^q \\\supp{M}\ni 0}} \frac{\abs{\varphi(M)}}{\abs{\supp{M}}}\prod_{\gamma\in M}\abs{w_{\star}^q(\gamma)}
        \le \frac{1}{\beta}\exp{-\frac{1}{6}\beta\tau\max\set[\bigg]{3^d,\bigg(\frac{\tau}{3a_{q}}\bigg)^{d-1}}},
    \end{equation}
    and the estimate~\eqref{eqn: bound on the finite-volume correction to the pressure} continues to hold with the auxiliary versions of the involved quantities.
    Hence, provided that
    \begin{equation}
        \label{eqn: convoluted assumption}
        \frac{1}{\beta}\exp{-\frac{1}{6}\beta\tau\max\set[\bigg]{3^d,\bigg(\frac{\tau}{3a_{q}}\bigg)^{d-1}}} \le a_q,
    \end{equation}
    we have that, for all sets $\Gamma\subset\cC^q(\Lambda)$ of compatible, external, large contours,
    \begin{align}
        {}& e^{-\beta a_{q}\abs{\ext{\Gamma}\cap\Lambda}} 
        \prod_{\gamma \in \Gamma} e^{-(\beta\tau-3)\abs{\supp{\gamma}}}
        {}&&\nonumber
        \\
        \le {}& e^{-\beta s_{q,\star}\abs{\ext{\Gamma}\cap\Lambda}} 
        \prod_{\gamma \in \Gamma} e^{-(\beta\tau-3)\abs{\supp{\gamma}}}
        {}&& \text{since $s_{q,\star}\le a_q$, by~\eqref{eqn: bound on the auxiliary pressure} and~\eqref{eqn: convoluted assumption}}
        \nonumber
        \\
        ={}& e^{-\beta s_{q,\star}\abs{\Lambda}} 
        \prod_{\gamma \in \Gamma} e^{-(\beta\tau-\beta s_{q,\star}-3)\abs{\supp{\gamma}}+\beta s_{q,\star}\abs{\interior{}{\gamma}}}
        {}&&
        \nonumber
        \\
        \le{}& e^{-\beta s_{q,\star}\abs{\Lambda}} 
        \prod_{\gamma \in \Gamma} e^{-(\beta\tau-4)\abs{\supp{\gamma}}+\beta s_{q,\star}\abs{\interior{}{\gamma}}}
        {}&& \text{since $\beta s_{q,\star}\le 1$, by~\eqref{eqn: bound on the auxiliary pressure}}
        \nonumber
        \\
        \le{}& e^{-\beta s_{q,\star}\abs{\Lambda}} 
        \prod_{\gamma \in \Gamma} e^{-(\beta\tau-5)\abs{\supp{\gamma}}} \Xi^q_\star(\interior{}{\gamma})
        {}&&\text{by~\eqref{eqn: bound on the finite-volume correction to the pressure}}
        \nonumber
        \\
        ={}& e^{-\beta s_{q,\star}\abs{\Lambda}} 
        \prod_{\gamma \in \Gamma} w^q_\star(\gamma) \Xi^q_\star(\interior{}{\gamma})
        {}&& \text{by Claim~\ref{itm: large implies geometrically large}}.
        \label{eqn: Pirogov-Sinai first long derivation}
    \end{align}
    Thus, 
    \begin{align}
        {}& \sum_{\substack{\Gamma\subset\cC^q(\Lambda) \\ \text{C.E.L.} }}
        e^{-\beta a_{q}\abs{\ext{\Gamma}\cap\Lambda}} 
        \prod_{\gamma \in \Gamma} e^{-(\beta\tau-3)\abs{\supp{\gamma}}}
        {}&&
        \nonumber
        \\
        \le {}& e^{-\beta s_{q,\star}\abs{\Lambda}} 
        \sum_{\substack{\Gamma\subset\cC^q(\Lambda) \\ \text{C.E.L.} }} 
        \prod_{\gamma \in \Gamma} w^q_\star(\gamma) \Xi^q_\star(\interior{}{\gamma})
        {}&&\text{by~\eqref{eqn: Pirogov-Sinai first long derivation}}
        \nonumber
        \\
        \le {}& e^{-\beta s_{q,\star}\abs{\Lambda}} 
        \Xi^q_\star(\Lambda)
        {}&& 
               \nonumber
        \\
        \le {}& e^{\abs{\exb\Lambda}}
        {}&& \text{by~\eqref{eqn: bound on the finite-volume correction to the pressure}},
    \end{align}
where the penultimate inequality follows since large contours form a subset of
geometrically large contours by Claim~\ref{itm: large implies geometrically large}.
This completes the proof.

\section{Crystallization}
\label{sec:crystallization}

In this section, we determine the low-temperature phase structure of every maximal hard-core model satisfying Assumptions~\ref{asm:optimization} and~\ref{asm:supercell}.
Informally, we will show that a typical maximal configuration at sufficiently low temperatures is a small perturbation of a perfect configuration.
The main ingredient is Theorem~\ref{thm:applyPS}, our formulation of the main technical result of Zahradn\'ik~\cite[Theorem~1.7]{zahradnik1984alternate}. 
Recall the reduction of the maximal model to an equivalent spin model via the common-supercell construction as described at the beginning of Section~\ref{sec:PS-theory}.
Recall also the notion of a \textit{stable} periodic ground state from Section~\ref{sec: truncated contour model}.

For a periodic ground state $\overline X^q$, an $\sF$-measurable region $\Lambda \Subset \Z^d$, and a maximal hard-core configuration $X$, define the finite-volume Gibbs measure with boundary condition $q$ as follows:
\begin{equation*}
    \P^q_{\Lambda,\beta} (X) 
    =\frac{e^{-\beta H_{\Lambda}(X_\Lambda  X_{\Lambda^c})}}{\bZ^q(\Lambda)}\indicator{\Omega_\Lambda^q}(X).
\end{equation*}

\begin{proposition}[{\cite[Lemma 7.26]{friedli2017statistical}}]\label{prop:FV7.12}
    Let $\gamma \in \mathcal{C}^q$ be a contour of type $q$ whose support is contained in $\Lambda \setminus \partial^{\mathrm{in}}\Lambda$.
Then
    \[
    \P^q_{\Lambda,\beta}(\gamma \text{ is external}) \le w^q(\gamma).
    \]
\end{proposition}
\begin{proof}
    Let $\cA_0(\Lambda)$ denote the collection of all subsets of $\cC^q(\Lambda)$ consisting of pairwise compatible contours, and let
    $\cA_1(\Lambda)\subset \cA_0(\Lambda)$ denote the subcollection of those $\Gamma\in\cA_0(\Lambda)$ with respect to which $\gamma$ is external.
    As in the derivation of~\eqref{eqn: contour representation}, we have
    \begin{equation*}
        \begin{split}
        \P^q_{\Lambda,\beta}(\gamma\text{ is external})
        &= \frac{\sum_{\Gamma\in\cA_1(\Lambda)}\prod_{\gamma\in\Gamma}w^q(\gamma)}
                {\sum_{\Gamma\in\cA_0(\Lambda)}\prod_{\gamma\in\Gamma}w^q(\gamma)}
        = w^q(\gamma)\,
          \frac{\sum_{\Gamma\in\cA_2(\Lambda)}\prod_{\gamma\in\Gamma}w^q(\gamma)}
               {\sum_{\Gamma\in\cA_0(\Lambda)}\prod_{\gamma\in\Gamma}w^q(\gamma)}
        \le w^q(\gamma),
    \end{split}
    \end{equation*}
    where $\cA_2(\Lambda)$ is the subcollection of all $\Gamma \in \cA_0(\Lambda)$ such that each $\gamma' \in \Gamma$ is compatible with $\gamma$ and there does not exist $\gamma' \in \Gamma$ such that $\supp{\gamma} \subset \interior{}{\gamma'}$; in other words, $\gamma$ can be added as an external contour to $\Gamma$, for each $\Gamma \in \cA_2(\Lambda)$.
\end{proof}

\begin{theorem}[Crystallization]
    \label{thm:volume allocation crystallization}
    Suppose that a maximal hard-core model satisfies
    Assumptions~\ref{asm:optimization} and~\ref{asm:supercell} and, in addition, a constant $\beta_0>0$ is chosen to satisfy Assumption~\ref{asm: parameters}.
    There exists a decreasing function $\epsilon:[\beta_0,\infty)\to[0,\infty)$ with $\lim_{\beta\to\infty}\epsilon(\beta)=0$ such that the following holds.
    For all $\beta \ge \beta_0$:
    \begin{enumerate}[label=(\roman*)]
        \item The set of stable periodic ground states is nonempty.
        \label{itm: stable PGS exists}
        \item The maximal hard-core model crystallizes around each stable periodic ground state $\overline X^q$; in other words,
        for every fixed cell $K$,
        the limit 
        \[
            \lim_{\Lambda \Uparrow \Z^d} \P^q_{\Lambda,\beta} \bigl(K\text{ is not $q$-correct}\bigr)
        \]
        exists and is bounded above by $\epsilon(\beta)$.
        \label{itm: q-correct likely in q-measure}
        \item For every stable periodic ground state
        $\overline X^{\,q}$, the finite-volume Gibbs measures with boundary condition $q$ converge to an $\cL^F$-invariant extremal Gibbs measure $\P_\beta^q$, and different stable periodic ground states give rise to different Gibbs measures.
        \label{itm: different stable PGS different measure}
  \end{enumerate}
\end{theorem}

\begin{proof}
    The construction preceding Assumption~\ref{asm:supercell} allows us to view the maximal hard-core model as an equivalent spin model on $\Z^d$ with a finite single-spin space and a translation-invariant, possibly infinite-valued interaction of range one (with respect to the sup-norm). 
    Proposition~\ref{prop:perfect iff periodic ground state} identifies the perfect configurations of the model with its periodic ground states, and by Assumption~\ref{itm:global optimization}, there are finitely many of them. 
  For this reason, the minimum defining $h_0$ in~\eqref{eqn: metastable free energies} is attained, so at least one $\overline X^q$ is stable, which proves \ref{itm: stable PGS exists}.
   Fix such a $q$.
 
 By Assumption~\ref{asm: parameters}, Theorems~\ref{thm:applyPS} and~\ref{thm:FV7.29} apply (the Peierls constant $\tau>0$
is given by Proposition~\ref{prop:Peierls condition}).
By Theorem~\ref{thm:applyPS}, every contour of type $q$ is stable, the contour representation~\eqref{eqn: contour representation} coincides with the truncated contour model~\eqref{eqn: truncated contour model}, and Theorem~\ref{thm:FV7.29} yields a convergent cluster expansion for the partition function with $q$-boundary condition.
    The cluster expansion yields
    convergence of the finite-volume Gibbs measures with $q$-boundary condition
    to an extremal, $\cL^F$-invariant Gibbs measure with exponential decay of
    correlations~\cite[Corollary on p.6, Parts~(iii), (iv)]{mazel2024pirogov}.
    We next show that there is a function
    $\epsilon(\beta)\to 0$
    such that the bound
    \begin{equation*} 
        \P_{\Lambda,\beta}^q(K\text{ is not }q\text{-correct})\le\epsilon(\beta).
    \end{equation*}
holds uniformly over every bounded, $\sF$-measurable region $\Lambda$, every cell $K\subset\Lambda$, and every stable periodic ground state $\overline X^q$.
     To see this, observe that there exists an external contour $\gamma \in \mathcal{C}^q$
    that separates $K$ from $\Lambda^c$ whenever $K$ is not $q$-correct with respect to a configuration in $\Omega^q_\Lambda$.
    Proposition~\ref{prop:FV7.12} and Theorem~\ref{thm:FV7.29}, together with
    the fact that $\gamma$ is stable, give
    \[  
        \begin{split}
            \P_{\Lambda,\beta}^q (K\text{ is not }q\text{-correct}) & \le \sum_{K \subset V(\gamma)} \P_{\Lambda,\beta}^q (\gamma \text{ is an external contour}) \le  \sum_{K \subset V(\gamma)} w^q(\gamma)\\
            &  \le \sum_{K \subset V(\gamma)} e^{-\beta\tau\abs{\supp{\gamma}}/3} =: \epsilon(\beta) \to 0
        \end{split}
    \]
    as $\beta \to \infty$, which proves~\ref{itm: q-correct likely in q-measure}. Let us stress that we
    are not taking the limit $\beta\to\infty$ on the left-hand side because stability of the
    periodic ground states, a priori, is not guaranteed for all $\beta$ in the limiting process.
    
    Finally, as soon as $\epsilon(\beta_0)<1/2$, distinct stable periodic ground states give rise to distinct Gibbs measures since, by the common-supercell construction, a cell cannot be both $p$-correct and $q$-correct for different periodic ground states $\overline X^p$, $\overline X^q$, which proves~\ref{itm: different stable PGS different measure}.
\end{proof}

A natural question that Theorem~\ref{thm:volume allocation crystallization} does not address is which of the periodic ground states are stable at each $\beta\ge\beta_0$.
Define a \textbf{symmetry} of a maximal hard-core model as a Euclidean isometry of $\R^d$ that preserves $\bG$ and induces an automorphism on the set of maximal configurations that preserves the chemical potential associated with each tile type.
It turns out that if a periodic ground state is stable at some $\beta$, then its images under symmetries of the model are also stable periodic ground states at $\beta$~\cite[Section~1.2]{mazel2025high}.
Hence, if the periodic ground states form a single orbit under symmetries of the model, then all are stable for all $\beta\ge\beta_0$, where $\beta_0$ is the same threshold as in the statement of Theorem~\ref{thm:volume allocation crystallization}.
Otherwise, one considers the distinct symmetry orbits of the periodic ground states and uses either the method implicit in Zahradn\'ik~\cite{zahradnik1984alternate} or that of Slawny~\cite{slawny1979low,slawny1987low} to distinguish the stable ones.
We do not belabor the details here and refer the reader to~\cite[Section~6]{mazel2025high} for a specific implementation of the former method for a hard-disk model (see Section~\ref{sec: hard disks}).

\section{Examples}\label{sec: Examples}

In this section, we collect a few examples of maximal hard-core models that satisfy Assumptions \ref{asm:optimization} and \ref{asm:supercell} and are therefore covered by our theory. 
Some of the results are collected in Table~\ref{tab:assumption one models} below.
We note that our formalism does not address the classification of periodic ground states, which is an NP-hard problem in general~\cite[p.~511]{friedli2017statistical}.
Instead, the classification must be manually accomplished, via Proposition~\ref{prop:perfect iff periodic ground state}, by producing a suitable volume allocation that distinguishes exactly finitely many configurations as perfect in the sense of Assumption~\ref{asm:optimization}, which, as the examples in this section will show, can be a formidable endeavor in itself.

\subsection{Nearest-neighbor maximal models}
\label{sec:nearest-neighbor-examples}

First, we study the nearest-neighbor maximal models on $\bA_2$, $\Z^d$ ($d\ge2$), and $\bH_2$ in both low- and high-activity regimes.

\subsubsection{The discrete volume allocation}\label{subsec:examples common allocation}

Let $\Delta$ denote the degree of the lattice (as a graph) and let $\alpha$ be the number of pairwise nonadjacent neighbors of a vertex.
Let $X$ be a maximal hard-core configuration and let $y \notin X$. 
Since the occupied neighbors of a vacant vertex form an independent set, one has
$1\le k_y(X)\le\alpha$.
Therefore,
\begin{equation}
    \label{eqn:examples allocated volume bound}
    1+\frac{\Delta}{\alpha}\le v(x\mid X) =1+\sum_{y\sim x}\frac{1}{k_y(X)} \le 1+\Delta.
\end{equation}
The upper bound is attained if and only if every vacant vertex has exactly one occupied neighbor, and the lower bound is attained if and only if every vacant vertex has $\alpha$ occupied neighbors.
We show that for all lattices considered in Sec.~\ref{sec:nearest-neighbor-examples}, both bounds are attainable.
The lower bound in~\eqref{eqn:examples allocated volume bound}
implies Assumption~\ref{itm:minorization} with
$p_\star=\alpha/(\alpha+\Delta)$ when $\sigma=\mu_x=1$,
while the upper bound implies it with
$p_\star=1/(1+\Delta)$ when $\sigma=\mu_x=-1$.

In the remainder of this section, we determine whether configurations attaining~\eqref{eqn:examples allocated volume bound} form a finite family. The results are summarized in Table~\ref{tab:assumption one models}.
To check whether Assumption~\ref{asm:supercell} holds we relied on Proposition~\ref{prop: locality implies supercell assumption}.

\begin{table}[tbp]
    \centering
    \small
    \renewcommand{\arraystretch}{1.25}
    \begin{tabular}{cccccccc}
        \toprule
        Model & $\sigma$ & $\Delta$ & $\alpha$ & $p_\ast$
        & \shortstack{number of perfect\\configurations} &\shortstack{chosen common\\ supercell}
        & Assumptions~\ref{asm:optimization}, \ref{asm:supercell} \\
        \midrule
        $\bA_2$ & $1$  & $6$ & $3$ & $1/3$ & 3    & $F_{\bA_2}(3)$   & hold \\
        $\bA_2$ & $-1$ & $6$ & $3$ & $1/7$ & 14   &$F_{\bA_2}(7)$ & hold \\
        $\Z^d$, $d\ge2$ & $1$ & $2d$ & $2d$ & $1/2$ & 2 & $[0,2)^d$  & hold \\
        $\Z^2$ & $-1$ & $4$ & $4$ & $1/5$ & 10 &  $[0,5)^2$   & hold \\
        $\Z^3$ & $-1$ & $6$ & $6$ & $1/7$ & 56 & $[0,7)^3$  & hold \\
        $\Z^5$ & $-1$ & $10$ & $10$ & $1/11$ & 4224 &$[0,11)^5$    & hold \\
        \shortstack{$\Z^d$, $d\ge4$\\$2d+1$ composite}
        & $-1$ & $2d$ & $2d$ & $1/(2d+1)$ & $\infty$ & none   & fail \\
        \shortstack{$\Z^d$, $d\ge6$\\$2d+1$ prime}
        & $-1$ & $2d$ & $2d$ & $1/(2d+1)$
        & $\le\binom{(2d+1)^d}{(2d+1)^{d-1}}$& $[0,2d+1)^d$  & hold \\
        $\bH_2$ & $1$  & $3$ & $3$ & $1/2$ & 2 & $F_{\bH_2}(3)$  & hold \\
        $\bH_2$ & $-1$ & $3$ & $3$ & $1/4$ & $\infty$ & none  & fail \\
        \bottomrule
    \end{tabular}
    \caption{Assumptions~\ref{asm:optimization}, ~\ref{asm:supercell} for the
    nearest-neighbor models.\\ Here $ F_{\bA_2}(r):=\set{s a_1+t a_2\mid 0\le s,t<r},
    F_{\bH_2}(r)
    :=\set{s b_1+t b_2\mid 0\le s,t<r}.$}
    \label{tab:assumption one models}
\end{table}

\subsubsection{Triangular lattice $\bA_2$}

{As shown in~\cite{wang2026recoverable}, there are $3$ perfect configurations for $\sigma = 1$ and $14$ for $\sigma = -1$. These attain the lower and upper bounds in~\eqref{eqn:examples allocated volume bound}, respectively.

There is another way
of defining a volume allocation for the same model, which is derived from
the analysis of the low-activity case in \cite[Section~3]{wang2026recoverable}. Let $X\in \sX(\bA_2)$ be a maximal independent set and let $\cT_X$ be a Delaunay triangulation
of the occupied sites embedded in $\R^2$. Denote by $\theta_x(T)$ the angle at the vertex $x$ of $T\in \cT_X$ and let $\cV(T)$ denote the vertices of $T$. 
Following Definition~\ref{def:volume allocation}, take $S=\R^2$ with $\lambda$ the Lebesgue measure and set
   \begin{equation}\label{eq:phi-ang}
  \phi_x(s\mid X):=\!\!\sum_{\substack{T\in\cT_X\\ \cV(T)\ni x}}
  \frac{\theta_x(T)}{\pi}\,\indicator{T}(s).
\end{equation}
It is easy to check that $\phi_x$ satisfies the conditions of a volume allocation. Translation invariance, continuity, and bounded support are straightforward. To check partition of unity,
   write
   \begin{equation}\label{eq: sum of angles}
  \sum_{x\in X}\phi_x(s\mid X)
  =\sum_{T\ni s}\frac{\sum_{x\in \cV(T)}\theta_x(T)}{\pi}=1
  \quad\lambda\text{-a.e.},
\end{equation}
which uses the fact that the sum of the angles is $\pi$ (this argument was used in 
\cite[Lemma~3.4]{wang2026recoverable}). For the allocated volume we obtain the expression
\begin{equation*}
    v(x\mid X)=\sum_{T\ni x} \frac{\theta_x(T)}{\pi}\abs{T}.
\end{equation*}
Using \eqref{eqn:local optimization} with $\sigma=-1$ and the fact that $\abs{T}\le \frac {7\sqrt 3}4$ for any $T\in \cT_X$ \cite[Section~3]{wang2026recoverable}, we obtain the inequality $v(x\mid X)\le \frac{7\sqrt 3}{2}$ with equality for the perfect configurations (this verifies Assumption 1).
This argument also implies that $p_\ast=\frac 2{7\sqrt 3}$, which can also be derived from the discrete volume allocation.

\subsubsection{The bipartite lattices at high activity.}
The lattices $\Z^d$, $d\ge2$, and $\bH_2$ are connected and bipartite,
with $\alpha=\Delta$.  Equality in the lower bound in \eqref{eqn:examples allocated volume bound} requires
$k_y(X)=\Delta$ for every vacant vertex $y$, so all its neighbors are occupied
and $\bG\setminus X$ is independent.  Since $X$ is also independent,
$(X,\bG\setminus X)$ must be the bipartition of the graph.  Hence
Assumption~\ref{asm:optimization} holds with $p_\ast=1/2$ and exactly two
perfect configurations.
Note that this argument uses only the fact that the lattice is bipartite and that $\alpha=\Delta$.

\subsubsection{The hypercubic lattices $\Z^d$ at low activity}\label{sec: perfect}
Here $\sigma=-1$ and $\alpha=\Delta=2d$, so $p_\ast=1/(2d+1)$ and~\eqref{eqn:examples allocated volume bound} reads $v(x\mid X)\le 2d+1$.  All $2d$ neighbors of an occupied vertex $x$
are vacant, so equality at every $x\in X$ means that every vacant vertex has exactly one
occupied neighbor.  In other words, the graph balls (also called \emph{crosses})
\begin{equation*}
    B^{\mathrm{gr}}_1(x)=\set{x,x\pm \vec{e}_1,\ldots,x\pm \vec{e}_d},
    \qquad x\in X,
\end{equation*}
partition $\Z^d$, where $\vec{e}_1,\ldots,\vec{e}_d$ are the standard basis vectors.  As the graph
distance on $\Z^d$ is the $\ell^1$, or Lee, distance, the perfect
configurations are precisely the radius-one perfect Lee codes, that is, the
sets of centers of a tiling of $\Z^d$ by the $(2d+1)$-site cross
$B^{\mathrm{gr}}_1(0)$~\cite{golombWelchPerfectCodes1970}, \cite[Section~9.2]{haynes2023domination}.
Such codes exist in every dimension.

If $2d+1$ is not prime, then there are $2^{\aleph_0}$ non-congruent non-periodic perfect configurations
\cite[Theorem~1(i)]{horakNonperiodicTilingsCrosses2012}, so
Assumption~\ref{asm:optimization} fails.

On the other hand, let $p := 2d+1$ be prime. Then all perfect configurations are periodic with period $p$ in every direction~\cite[Corollary~35]{horakKimFiftyYears2018}. Denote by $N_d$ the number of perfect configurations. Since each cell of $p^d$ sites is fully packed by $p^{d-1}$ crosses of size $p$, we have
\begin{equation*}
    N_d \le \binom{p^d}{p^{d-1}} < \infty
\end{equation*}
and Assumption~\ref{asm:optimization} holds.

To justify the count of perfect configurations in $\Z^d$ listed in Table~\ref{tab:assumption one models}, we briefly recall a classical algebraic construction of  ~\cite[Theorem~3]{golombWelchPerfectCodes1970}. Call $\vec a = (a_1,\dots, a_d)\in\bigl(\Z/p\Z\bigr)^d$ \emph{admissible} if 
\begin{equation}\label{eqn:admissible}
    \set{0,\pm a_1,\ldots,\pm a_d} = \Z / p \Z.
\end{equation}
For admissible $\vec a$ and $c\in\Z/p\Z$, put
\begin{equation}\label{eqn:perfect-Zd-tiles}
    \overline X^{(\vec a,c)}
    :=\set{\vec x\in\Z^d\mid\vec a\cdot\vec x\equiv c\pmod{p}}.
\end{equation}
To see that $\overline X^{(\vec a,c)}$ is perfect, fix $\vec z\in\Z^d$.
By~\eqref{eqn:admissible}, there is a unique
$\vec{e}\in B^{\mathrm{gr}}_1(0)=\set{0,\pm \vec{e}_1,\ldots,\pm \vec{e}_d}$ such that
\begin{equation*}\label{eqn:unique-e}
    \vec a\cdot \vec{e}\equiv \vec a\cdot\vec z-c\pmod{p}.
\end{equation*}
Then $\vec x=\vec z-\vec{e}$ belongs to $\overline X^{(\vec a,c)}$, and the uniqueness
of $\vec{e}$ shows that $\vec x$ is the unique code point within distance one of
$\vec z$.  Thus $\overline X^{(\vec a,c)}$ is a perfect configuration.

There are $2^d d!$ choices for $\vec a$. Indeed, the $2d$ nonzero residues
modulo $2d+1$ form $d$ opposite pairs $\set{r,-r}$.  These pairs can be
assigned to the ordered coordinates of $\vec a$ in $d!$ ways, after which one
of the two representatives in each pair can be chosen in $2^d$ ways.

Since $p$ is prime, $\overline X^{(\vec b,c')}
    =\overline X^{(\vec a,c)}$ if
and only if $(\vec b,c')=(t\vec a,tc)$
for some $t\in\bigl(\Z/p\Z\bigr)^\times$.
Thus each sublattice
arising from this construction is represented by exactly $2d$ admissible
vectors, and the number of distinct sublattices is
\begin{equation}\label{eqn:perfect-count}
    \frac{2^d d!}{2d}=2^{d-1}(d-1)!.
\end{equation}
Each sublattice has $2d+1$ cosets, so the construction gives $p2^{d-1}(d-1)!$ perfect configurations.

In particular, for $d=2,3$~\cite[Theorem~2]{horakNonperiodicTilingsCrosses2012} and $d=5$~\cite[Theorem~4]{horakTilingR5Crosses2014}, any two perfect configurations are congruent.
Thus every perfect
configuration in these dimensions is one of the lattice cosets defined by~\eqref{eqn:perfect-Zd-tiles}. 
By~\eqref{eqn:perfect-count}, we have
\begin{equation*}
    N_2=10,\qquad N_3=56,\qquad N_5=4224.
\end{equation*}
The case of $d=2$ was analyzed from first principles in \cite{jauslin2018high,wang2025maximal}.
For $p\ge 13$ prime, it is conjectured (see~\cite[Conjecture~3]{horakNonperiodicTilingsCrosses2012}) that this construction exhausts all perfect configurations.

\begin{example}
We exemplify the use of Theorem~\ref{thm:volume allocation crystallization} for maximal independent sets in $\Z^2$ in the case of low activity. To find the bound for
 activity $u$ below which this theorem applies, we solve~\eqref{eqn: parameters 1}--\eqref{eqn: parameters 3} of Assumption~\ref{asm: parameters}. It is shown in~\cite[Theorem~6.4]{wang2025maximal} that we can take the Peierls constant $\tau=1/270$. Further, we find by computer that the number of distinct restrictions
 of maximal independent sets to a $5\times 5$ square is $\kappa=4883$. The dominant inequality among the three conditions is the one in \eqref{eqn: parameters 2}. The series in it converges if $\beta \ge 22725$, which 
 gives a bound on $\beta_0$. 
 Recalling that $u=e^{\mu \beta}$, we conclude that 
crystallization occurs for all $u\le \exp(-22725)$.
 \end{example}

\subsubsection{Non-example: the honeycomb lattice $\bH_2$ at low activity}\label{sec: hexagonal}
Here $\Delta=3$, so $v(x\mid X)\le4$ and $p_\ast=1/4$.  As on $\Z^d$, equality
holds at every $x\in X$ when $X$ is a
radius-one perfect code, or an efficient dominating set
of~$\bH_2$~\cite{zhaoPerfectDominationRatios2022,senthilthilakEfficientDominationLattice2023}.
We show that there are infinitely many periodic ones.

Group the vertices of $\bH_2$ into the pairs $\set{x,x+\delta}$,
$x\in\bL$, and the pairs into the rows $x+\Z b_1$, successive rows
differing by $b_2$.  
In every second row, occupy one vertex of each pair,
alternating between the two translates along the row, and leave the remaining
rows empty: the occupied vertices lie on zig-zag lines whose closed
neighborhoods tile $\bH_2$, so the configuration is perfect.  Each line may
sit in either of two positions, which differ by $b_1$, and its position
constrains none of the others; every line therefore slides on its own
(Figure~\ref{fig:honeycomb low activity sliding}), and any choice of positions
is again perfect.  There are therefore infinitely
many periodic perfect configurations:
sliding one line in every $\ell\ge3$ gives a configuration whose period lattice
is generated by $2b_1$ and $2\ell b_2$, of index $4\ell$ in $\bL$. Hence Assumption~\ref{itm:global optimization} fails. 

\begin{figure}[tbp]
    \centering
    \includegraphics[width=0.8\textwidth]{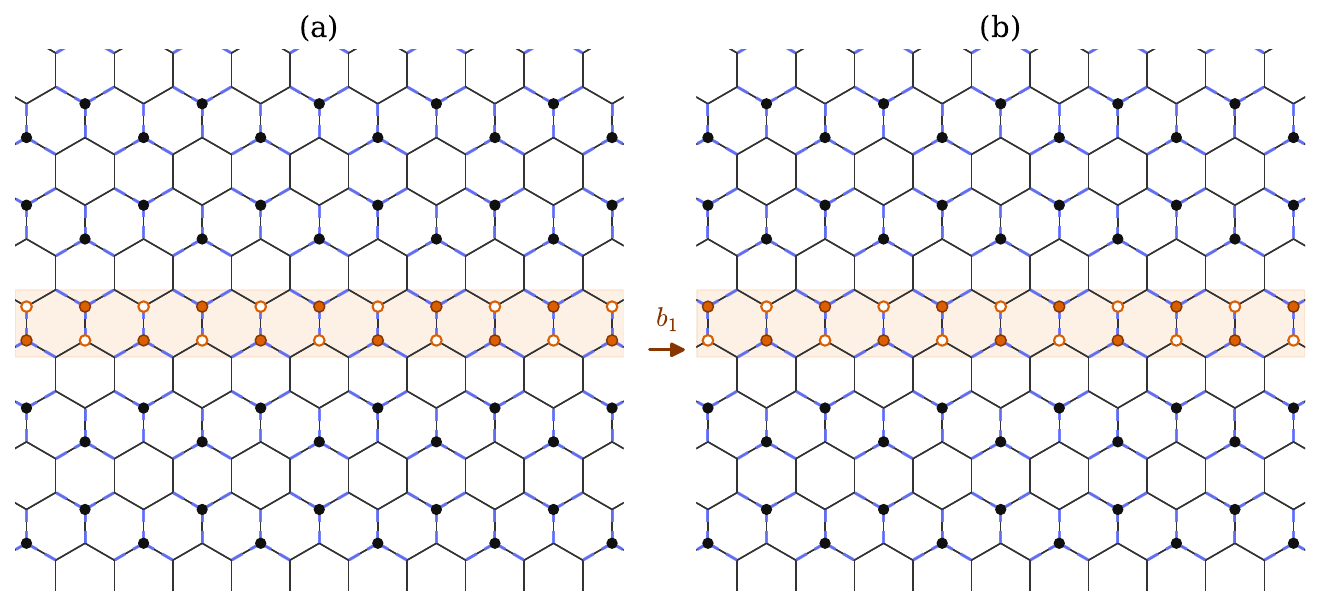}
    \caption{The sliding line.  Both panels are perfect configurations of
    $\bH_2$ at low activity: The panels differ in the
    position of the shaded line alone, which (b) slides by $b_1$. 
    Any set of lines may be slid in
    this way; sliding one line in every $\ell$ gives a periodic configuration
    of period $2\ell b_2$.}
    \label{fig:honeycomb low activity sliding}
\end{figure}

\subsection{Beyond nearest-neighbor models}\label{sec:beyond NN}

\subsubsection{A two-tile example}

We consider a two-tile example on $\bG = \bL =  \Z^2$ with 
\begin{align*}
    T_1 &= B_\epsilon((0,0)) \cup B_\epsilon((1/5,1/5)) \cup B_\epsilon((2/5,2/5)) \cup B_\epsilon((2/5,7/5)),\\
    T_2 &= B_\epsilon((0,0)) \cup B_\epsilon((6/5,1/5)) \cup B_\epsilon((3/5,3/5)) \cup B_\epsilon((3/5,8/5)),
\end{align*}
where $0<\epsilon < 1/100$ is a small constant such that two tiles overlap if and only if they coincide on a disk.
See Figure~\ref{fig:two tile base tiles}.

\begin{figure}[tbp]
    \centering
    \begin{subfigure}[b]{0.25\textwidth}
        \centering
        \includegraphics[width=.8\textwidth]{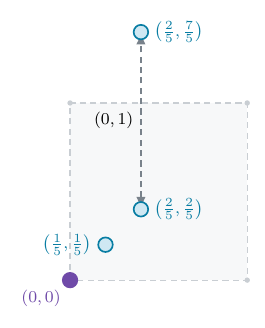}
        \caption{tile $T_1$}
        \label{fig:two tile T1}
    \end{subfigure}
    \hspace{0.05\textwidth}
    \begin{subfigure}[b]{0.35\textwidth}
        \centering
        \includegraphics[width=\textwidth]{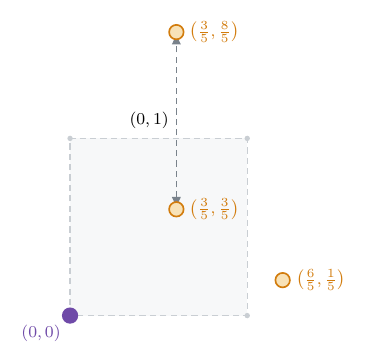}
        \caption{tile $T_2$}
        \label{fig:two tile T2}
    \end{subfigure}
    \caption{The tiles $T_1$ and $T_2$, each drawn as its four disks of radius
    $\epsilon$, with the centers labeled.  The dashed square is the unit cell
    of $\bG=\Z^2$ and the gray dots are its sites.  The disk at the origin
    (purple) is common to both tiles, and the dashed arrow joins the two disks
    of a tile whose centers differ by $(0,1)$.}
    \label{fig:two tile base tiles}
\end{figure}

For $i \in \set{1,2}$ and $(a,b) \in \Z^2$, we write
\begin{equation*}
    T^{(i)}_{a,b} := (a,b) + T_{i}.
\end{equation*}
To study the maximal hard-core model defined by $T_1$ and $T_2$, we construct an equivalent model on $\tG \cong \Z^2$ as follows.

Let $X$ be a configuration on $\bG$.  We map $X$ to a configuration on $\tG$ by the map $\Phi$ defined as follows: for every $(a,b) \in \bG$, if $T^{(1)}_{a,b} \in X$, then occupy $(2a,b) \in \tG$; if $T^{(2)}_{a,b} \in X$, then occupy $(2a+1,b) \in \tG$.

\begin{proposition}
    \label{prop:two tile bijection}
    Let $\Omega_\bG$ and $\Omega_{\tG}$ be the collections of maximal hard-core configurations on $\bG$ and on $\tG$, respectively.  Then $\Phi$ is a bijection between $\Omega_\bG$ and $\Omega_{\tG}$.
\end{proposition}

\begin{proof}
    Two tiles overlap when they share a disk; there are four possibilities:
    $T^{(1)}_{a,b}$ with $T^{(2)}_{a,b}$, $T^{(2)}_{a,b}$ with $T^{(1)}_{a+1,b}$, $T^{(1)}_{a,b}$ with $T^{(1)}_{a,b+1}$, and $T^{(2)}_{a,b}$
    with $T^{(2)}_{a,b+1}$.
    Under $\Phi$ these four exclusions read:
    $(2a,b)$ and $(2a+1,b)$ are not both occupied;
    $(2a+1,b)$ and $(2a+2,b)$ are not both
    occupied;
    $(2a,b)$ and  $(2a,b+1)$ are not both occupied;
    $(2a+1,b)$ and  $(2a+1,b+1)$ are not both occupied.
    That is, $\Phi(X)$ is hard-core on $\tG$ if and
    only if $X$ is hard-core on $\bG$.  As $\Phi$ is a bijection between
    hard-core configurations and preserves inclusion, it maps $\Omega_\bG$ onto
    $\Omega_{\tG}$.
\end{proof}

\begin{corollary}
    Assumptions~\ref{asm:optimization} and~\ref{asm:supercell} hold for this model at both high and low activity, and Theorem~\ref{thm:volume allocation crystallization} applies.
\end{corollary}

\begin{proof}[Proof sketch]
    By Proposition~\ref{prop:two tile bijection} the two models differ by a relabeling only. 
    Take $\vec{\mu}=\frac{\sigma}{\sqrt2}(1,1)$ and allocate to a tile
    of $X$ half the volume that the allocation of
    Example~\ref{ex:Z2 counting allocation} gives to the corresponding
    particle of $\Phi(X)$.  Assumptions~\ref{asm:optimization} and~\ref{asm:supercell} then hold, and the perfect
    configurations on $\bG$ are given by those of the nearest-neighbor maximal
    model on $\tG \cong \Z^2$ (Figure~\ref{fig:two tile perfect phases}).
\end{proof}
\begin{figure}[tbp]
    \centering
    \begin{subfigure}[b]{0.285\textwidth}
        \centering
        \includegraphics[width=\textwidth]{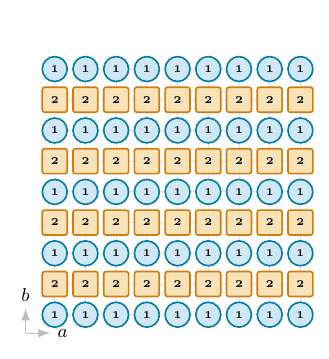}
        \caption{high activity}
        \label{fig:two tile phase high}
    \end{subfigure}
    \hfill
    \begin{subfigure}[b]{0.285\textwidth}
        \centering
        \includegraphics[width=\textwidth]{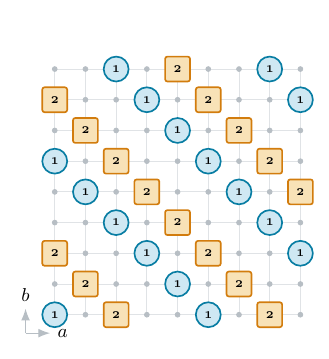}
        \caption{low activity}
        \label{fig:two tile phase low a}
    \end{subfigure}
    \hfill
    \begin{subfigure}[b]{0.285\textwidth}
        \centering
        \includegraphics[width=\textwidth]{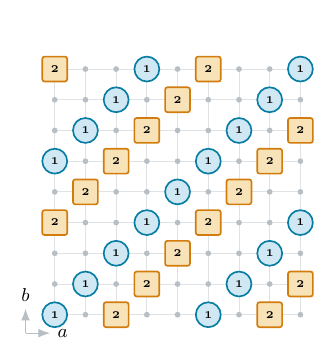}
        \caption{low activity}
        \label{fig:two tile phase low b}
    \end{subfigure}

    \vspace{2ex}
    \begin{subfigure}[b]{0.21\textwidth}
        \centering
        \includegraphics[width=\textwidth]{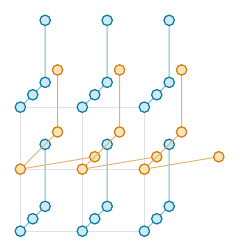}
        \caption{the tiles of (a)}
        \label{fig:two tile real high}
    \end{subfigure}
    \hspace{0.07\textwidth}
    \begin{subfigure}[b]{0.32\textwidth}
        \centering
        \includegraphics[width=\textwidth]{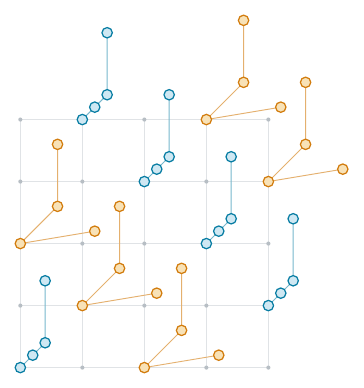}
        \caption{the tiles of (b)}
        \label{fig:two tile real low}
    \end{subfigure}
    \caption{The perfect configurations of the two-tile model on $\bG$.  A
    circle labeled $1$ marks a site carrying $T_1$, a square labeled $2$ marks a
    site carrying $T_2$, and a gray dot marks a vacant site.  (a): at high activity
    every site is occupied and the type alternates from row to row; the second
    phase is the translate of this one by $(0,1)$.  (b), (c): at low activity
    the occupied sites carry the two orientations of the perfect code of
    $\tG$; the remaining eight phases are translates of these two.  (d),
    (e): a window of the phases in (a) and (b) drawn with the actual tiles,
    the segments joining the four disks of a tile being guides only.}
    \label{fig:two tile perfect phases}
\end{figure}

\begin{remark}
    The number of tile types is therefore a feature of how a model is presented: two models with different numbers of tile types can be combinatorially identical.
\end{remark}

\subsubsection{Hard-disk models}
\label{sec: hard disks}
High-activity crystallization in hard-disk models on $\Z^2$, $\bH_2$, and $\bA_2$ without the maximality constraint is studied in a series of papers by Mazel, Stuhl, and Suhov~\cite{mazel2019high, mazel2020hard, mazel2025high}.
Admissible configurations in these works are governed by the condition that disks of diameter $D$ centered at the lattice sites do not overlap, and the probabilities are given by the usual hard-core measure \eqref{eq: hardcore measure}. 
For hard-disk models on $\bH_2$ and $\bA_2$, the procedure of \textit{area redistribution} introduced in~\cite[Section~4.2]{mazel2025high} yields a volume allocation related to the angular volume allocation in~\eqref{eq:phi-ang}.

Starting with a Delaunay triangulation $\cT_X$ of the configuration $X$, the authors of \cite{mazel2025high} partition its triangles into groups $C$ of sizes 2, 3, and 4. Recasting their
arguments in our terminology, they consider the volume allocation
   $$
   \phi_x(s\mid X)=\sum_{\substack{T\in \cT_X\\ T\ni x}}\frac{\theta_x(T)}\pi \frac{\indicator{\bigcup_{T'\in C}T'}(s)}{|C|},
  $$
  where $C$ is the group such that $T\in C$.
Then the allocated volume of the site $x$ is given by
   $$
v(x\mid X)=\sum_{T\ni x} \frac{\theta_x(T)}\pi {\mathscr{S}}(T,X),
   $$
where ${\mathscr{S}}(T,X):=\abs{\bigcup_{T'\in C}T'}/\abs{C}$ is the {\em average area} of a triangle in the group that contains $T$.
They further define the {\em minimal redistributed area} of the configuration $X$ and call a configuration perfect if the area of all the triangles in $\cT_X$ equals this value. They finally prove in \cite[Lemma~4.3]{mazel2025high} that any periodic ground state is a perfect configuration.

The proof that $\phi_x(s\mid X)$ satisfies the conditions for a volume allocation in Definition~\ref{def:volume allocation} is not difficult; for instance,
partition of unity is shown as in~\eqref{eq: sum of angles}.
Moreover, the notion of a perfect configuration in these papers coincides with ours as specialized to the disk-packing model. 
Hence, upon verifying Assumption \ref{asm:optimization} \textit{and} showing that there are finitely many perfect configurations, we could claim high-activity crystallization in these models with the maximality constraint.
However, the finiteness claim is nontrivial and does not hold for all values of the diameter $D$ due to the so-called \textit{sliding} phenomenon; see \cite[Sections~1.3,~1.4]{mazel2025high}.

\subsubsection{Polyforms with a finite number of tilings}

Informally, a polyform is \textit{a plane figure or solid compound consisting of multiple connected copies of a given base shape}~\cite{weisstein_polyform}.
Common examples of polyforms in two dimensions include polyominoes, polyiamonds, and polyhexes, respectively obtained from $\Z^2$, $\bA_2$, and $\bH_2$ by gluing together faces of these lattices along their shared edges.
Polyforms have long been a popular object of study in combinatorics~\cite{martin1991polyominoes,golomb1996polyominoes}, in part due to their wide-ranging packing and tiling behaviors, and have found applications in disciplines far beyond mathematics.
Polyomino fluids and mixtures, for instance, were studied by a group of computational chemists~\cite{barnes2009structure,barnes2010development} seeking to understand their self-assembly properties.

Polyforms naturally give rise to maximal hard-core models by choosing $\bG=\bL$ to be the $d$-dimensional lattice from which a polyform $T$ was constructed, and the tiles $T_1,\dots,T_k$ to be the translationally inequivalent images of $T$ under a subset of symmetries of $\bL$, often including its rotations and sometimes its reflections as well.
Here, we take $\mu_1=\dots=\mu_k=1/\sqrt{k}$.

To study these polyform models uniformly, it is convenient to use the volume allocation induced by the continuous Voronoi tessellation, defined as follows.
Let $S=\R^d$ and $\lambda$ be the Lebesgue measure.
For a maximal polyform configuration $X$ and $x\in X$, define the Voronoi cell associated with the tile $T_x$ by
\begin{equation}
    V_X(T_x):=\set{s\in\R^d\mid \rmd(s,T_x)=\min_{y\in X}\rmd(s,T_y)},
\end{equation}
where $\rmd$ denotes the usual Euclidean distance, and set
\begin{equation}
    \phi_x(s\mid X):=\frac{\indicator{V_X(T_x)}(s)}{\sum_{y\in X}\indicator{V_X(T_y)}(s)}.
\end{equation}
This is a direct, continuous analogue of the discrete volume allocation introduced in Definition~\ref{def: discrete allocation}.
By a similar argument, it satisfies the properties enumerated in Definition~\ref{def:volume allocation} and hence gives rise to a volume allocation.

With this volume allocation, it follows from the assumption that the polyform $T$ has finitely many tilings that its corresponding maximal model, in the high-activity regime, satisfies Assumption~\ref{asm:optimization} with $p_\ast=1/(\sqrt{k}\lambda(T))$, and the perfect configurations of the model are given precisely by the tilings.
Moreover, a similar argument to that for Corollary~\ref{cor: nn model assumption implication} shows that the volume allocation is local at each perfect configuration, so Proposition~\ref{prop: locality implies supercell assumption} yields Assumption~\ref{asm:supercell}.
Hence, Theorem~\ref{thm:volume allocation crystallization} applies and proves crystallization at high activities, in agreement with the established high-activity behavior of these hard-core models without the maximality constraint~\cite[Corollary~3.3 and Remark~3.4]{he2026unified}.

\section*{Acknowledgments} The authors are grateful to Alex Mazel for helpful discussions.
A.B. and G.W. were partially supported by NSF grants CCF2330909 and CCF2526035.
Q.H. was supported by an SAS Fellowship as a PhD student at Rutgers University when this work began and in part through a summer coadjutant appointment, made possible by a gift from Prof. Joel Lebowitz, after his graduation.


\end{document}